\documentclass[12pt]{article}
\usepackage{amsmath,amsfonts,amssymb,amsthm}
\usepackage{hyperref}

\def\be{\begin{equation}}
\def\ee{\end{equation}}
\def\ba{\begin{array}}
\def\ea{\end{array}}
\def\dps{\displaystyle}

\def\bpsi{\bar{\psi}}

\newcommand{\half}{\frac{1}{2}}

\def\sm{s_1-s_2-1}

\baselineskip
\advance\textheight\topskip
\renewcommand{\tilde}{\widetilde}
\renewcommand{\hat}{\widehat}
\newtheorem{prop}{Proposition}[section]

\newcommand{\bref}[1]{\textbf{\ref{#1}}}

\def\ads{AdS_{5}}
\def\be{\begin{equation}}
\def\ee{\end{equation}}
\def\ba{\begin{array}}
\def\ea{\end{array}}
\def\d{\partial}
\def\dps{\displaystyle}

\def\ba{\begin{array}}
\def\ea{\end{array}}
\def\d{\partial}
\def\dps{\displaystyle}

\def\bpsi{\bar{\psi}}

\newcommand{\dd}{\partial}
\renewcommand{\d}{\partial}

\renewcommand{\geq}{\,{\geqslant}\,}
\renewcommand{\leq}{\,{\leqslant}\,}

\newcommand{\binner}[2]{%
  {\langle}\kern-4.15pt{\langle}#1{,}\,#2{\rangle}\kern-4.15pt{\rangle}}

\newcommand{\ffrac}[2]{\raisebox{.5pt}%
  {\footnotesize$\displaystyle\frac{#1}{#2}$}\kern1pt}

\newcommand{\dl}[1]{\mathchoice{\ffrac{\dd}{\dd #1}}{\frac{\dd}{\dd
      #1}}{\ffrac{\dd}{\dd #1}}{\ffrac{\dd}{\dd #1}}}

\def\cA{\mathcal{A}}
\def\cB{\mathcal{B}}

\def\cD{\mathcal{D}}
\def\cE{\mathcal{E}}
\def\cF{\mathcal{F}}

\def\cM{\mathcal{M}}
\def\cN{\mathcal{N}}
\def\cO{\mathcal{O}}
\def\cP{\mathcal{P}}

\def\cR{\mathcal{R}}

\numberwithin{equation}{section} \makeatletter
\@addtoreset{equation}{section}

\begin{document}

\begin{flushright}
FIAN-TD-2010-09 \\
\end{flushright}

\vspace{1cm}

\begin{center}

{\Large\textbf{
FV-type action for $AdS_5$  mixed-symmetry fields}}

\vspace{.9cm}

{\large Konstantin Alkalaev}

\vspace{0.5cm}

\textit{I.E. Tamm Department of Theoretical Physics, \\P.N. Lebedev Physical
Institute,\\ Leninsky ave. 53, 119991 Moscow, Russia}

\vspace{0.5cm}

\begin{abstract}
We formulate Fradkin-Vasiliev type theory  of massless higher spin
fields in $AdS_5$. The corresponding  action functional  describes
cubic order  approximation to gravitational interactions of
bosonic mixed-symmetry fields of a particular "hook" symmetry type
and totally symmetric bosonic and fermionic fields.

\end{abstract}

\end{center}


\section{Introduction}

Interacting theories with spectra including graviton along with particles of spin grater than two  provide
a fascinating playground for exploring the gravity both on classical and quantum
levels. For example,  string theory describes
a dynamics of an infinite collection
of massive fields with growing masses and spins and a finite set of massless lower spin fields.
An important  feature of  higher spin models is  infinite
symmetries which are believed to improve conventional quantum inconsistency
of Einstein gravity.
Higher spin theories with massless spectra play a distinguished role because they
can be considered as an unbroken phase for  massive higher spin theories including
 string theory itself \cite{Vasiliev:1999ba,Gross:1988ue}
(see also, \textit{e.g.},
\cite{Sundborg:2000wp,Klebanov:2002ja,Sezgin:2002rt,Bonelli:2003kh,Bianchi:2004xi,Giombi:2009wh,
Henneaux:2010xg,Gaberdiel:2010pz,Koch:2010cy,Douglas:2010rc}
for a discussion in the AdS/CFT correspondence context).

The problem of constructing a consistent theory of
interactions between higher spin massless fields and the
gravity has been first  attacked  by Aragone and Deser
\cite{Aragone:1979hx}. According to them massless fields of  spin $s>2$ do not
minimally interact with the gravity and therefore no higher spin extension of
supergravity theories is possible (see
\cite{Sorokin:2004ie, Bekaert:2010hw} for a review). The solution
has been proposed by Fradkin and Vasiliev in
\cite{Fradkin:1987ks,Fradkin:1986qy} who formulated guiding
principles to construct a consistent interacting theory of higher
spin fields. They identified  anti-de Sitter background geometry as a natural background for
gravitational higher spin interactions and explicitly constructed
higher spin gauge symmetry algebra \cite{Fradkin:1986ka}. It turns out that
the presence  of additional dimensionful parameter -- the cosmological constant
$\lambda$ of anti-de Sitter spacetime -- enables  one to build
various higher derivative interaction terms in the action with
overall coefficients proportional to the inverse of $\lambda$,
and this is quite similar to string theory vertices of
massive higher spin fields.\footnote{In particular, it implies
that the straightforward $\lambda  \rightarrow 0 $ limit is
ill-defined thereby conforming the no-go theorem of
\cite{Aragone:1979hx}. However there exists a tricky limiting
procedure that allows one to build some non-minimal couplings of
higher spin fields with the gravity \cite{Boulanger:2008tg}.
See also recent papers \cite{Zinoviev:2008ck} which consider
some particular vertices of spin-$3$ massless field with the gravity. Moreover, using the analogy
between massless fields in AdS spacetime and massive fields in Minkowski space
these results are extended to interacting massive spin-$3$ fields in Minkowski spacetime \cite{Zinoviev:2008ck}. }
Let us mention that a wide class of higher spin cubic (self-)interaction vertices is known
in Minkowski space but they do not however contain minimal couplings
with the gravity
\cite{Berends:1984wp,Bengtsson:1986kh,Deser:1990bk,Fradkin:1991iy,Bekaert:2005jf,
Boulanger:2005br,Metsaev:2005ar,Manvelyan:2010jr}.

More recently the original FV theory has been extended from $d=4$
to $d=5$  for both $\cN=0$ pure bosonic  and
$\cN=1$ supersymmetric cubic interactions of totally symmetric
(Fronsdal) fields \cite{Vasiliev:2001wa, Alkalaev:2002rq}. The
$5d$ theory inherits all basic features of the $4d$ theory and is
governed by the higher spin symmetry superalgebra identified by Fradkin
and Linetsky in the context of the $4d$ higher spin conformal
theory \cite{Fradkin:1989yd,Fradkin:1989md}.\footnote{This algebra was
also identified as an algebra of global $\ads$ HS symmetries within $5d$ unfolded formulation proposed in
\cite{Sezgin:2001zs} .  More general class of
conformal higher spin algebras has been  described  in
\cite{Vasiliev:2001zy}. } The novel feature as compared to $4d$ FV
theory is an infinite degeneracy of the spectrum of  excitations:
a field of each spin enters in an infinitely many copies. In this
respect the spectrum of $5d$ FV-type theory  resembles
that of string theory where massive excitations of a given
spin appear on different mass levels growing up to infinity.

Going to higher dimensions  one encounters a new phenomenon
though: there are more than one spin number in $d>4$  so fields of
mixed-symmetry type described by $o(d-1)$ Young diagrams appear.
Mixed-symmetry $AdS_d$  fields  may interact to each other and
with totally symmetric fields including gravity so it will be
interesting to study their interactions. In particular, a FV-type
theory for mixed-symmetry fields is still unknown.\footnote{ The
cubic interaction vertices of mixed-symmetry fields in Minkowski
spacetime were analyzed within the light-cone formalism in
\cite{Fradkin:1995xy}. Inspired by string field theory some
covariant vertices for mixed-symmetry fields in Minkowski space
were constructed in \cite{Fotopoulos:2007nm}.}
In this paper we partially fulfill this gap and explicitly
construct cubic order interacting theory in $\ads$ that includes
mixed-symmetry field vertices.

We build $\cN=2$  FV-type theory thereby extending  $\cN=0$ and $\cN=1$ results obtained
previously \cite{Vasiliev:2001wa,Alkalaev:2002rq}. The higher spin algebra
that governs consistent interactions in our model is $\cN=2$ Fradkin-Linetsky
superalgebra \cite{Fradkin:1989yd,Fradkin:1989md}. It contains
$\cN=2$ extended $su(2,2|2)$ superalgebra  as a maximal finite-dimensional subalgebra so fields
of the theory are organized in $su(2,2|2)$ supermultiplets. Obviously, $AdS_5$ symmetry algebra $su(2,2)$
and $R$-symmetry algebra $u(2)$ are bosonic subalgebras of  $su(2,2|2)$ superalgebra. Contrary
to spectra of  $\cN=0,1$ theories the $\cN=2$ supermultiplet
contains not only totally symmetric fields but also the so-called
"hook" fields. The "hooks" are mixed-symmetry fields with particular
symmetry type differing from totally symmetric fields by
additional row of a single cell in the respective Young diagram. Denoting spins of
$\ads$ massless gauge fields by a pair of (half-)integer numbers $(s_1,s_2)$ we give the content
of $\cN=2$ spin-$s$ supermultiplet
\be
\label{1.1}
\{s \} = (s, 0)_{[1]}\oplus (s-\half)_{[2]}\oplus (s-1)_{[4]}\oplus (s-1,1)_{[1]}\oplus(s-\frac{3}{2})_{[2]}
\oplus (s-2)_{[1]}\;\;,
\ee
where $s$ is a highest spin, while labels in square brackets denote dimensions of
$u(2)$ algebra representations. Each spin-$s$ supermultiplet possesses equal number of $16s-8$
bosonic and fermionic degrees of freedom.

Generally, $AdS_5$ higher spin models based on Fradkin-Linetsky superalgebra
describe an infinite collection of  supermultiplets \eqref{1.1} with
a highest spin $s$ running up to infinity
\be
\sum_{k=0}^L \sum_{s=2}^\infty \;\;\{s\}^{(\,k\,)}
\ee
while  $k$ parameterizes a $k$-th copy of a spin-$s$ supermultiplet. The models considered
in this paper corresponds to $L=\infty$ (unreduced model) or $L=0$ (reduced model).

According to \eqref{1.1} the spectrum of massless excitations in a full $\cN=2$ supersymmetric theory includes lower spins
$0, \half, 1$  contained  in the spin-$2$ (graviton) and spin-$3$ (hypergraviton)
supermultiplets. However, we eliminate all these lower spin fields so that the resulting  theory
is not supersymmetric in a strong sense, \textit{i.e.} it is not globally supersymmetric.
It is legitimate  because in the cubic approximation
one can set to zero a coupling of any three fields keeping the gauge invariance
of the theory intact.\footnote{Indeed, a spin $s_1$-$s_2$-$s_3$ cubic coupling can be represented as
$g\, \Phi^{(s_1)}_{a_1 ... a_{s_1}} J^{a_1 ... a_{s_1}}(\Phi^{(s_2)}, \Phi^{(s_3)})$, where $g$ is a
coupling constant, $\Phi^{(s_i)}$ are spin-$s_i$ fields,
and $J^{a_1 ... a_{s_1}}(\Phi^{(s_2)}, \Phi^{(s_3)})$ are higher spin currents
bilinear in the fields and their derivatives. Gauge invariance of the above
coupling  implies that the currents
are conserved
$\cD_{a_1}J^{a_1a_2 ... a_{s_1}}(\Phi^{(s_2)}, \Phi^{(s_3)})\approx 0$, where $\approx$
means going on-shell while in the cubic order approximation
it is sufficient to use free field equations for $\Phi^{(s_{1,2})}$.
Recall also  that Jacobi identities of the gauge algebra
are proportional to $g^2$.
As a result,
gauge symmetry do not mix  different cubic couplings and one can consistently
switch off any of them. } This allows one to truncate all  vertices with lower spin fields which is
equivalent to eliminating them from the spectrum. It greatly simplifies the whole analysis
because within the FV-type theory the
action functionals for lower spin and higher spin fields are formulated in different terms
thereby leading to  some technical complications
(see, however, \cite{Fradkin:1986qy,Bekaert:2009ud}).

The paper is organized as follows. In Section \bref{sec: free} we
extensively discuss  the unfolded formulation of higher spin
dynamics in the $\ads$ background geometry in spinor language. We
consider totally symmetric fields with integer and half-integer
spins and mixed-symmetry fields of the  "hook" symmetry
type. The respective set of unfolded fields is given by
physical, auxiliary and extra fields which play  different dynamical
roles. We build quadratic Lagrangians  and
introduce the set of constraints that express all auxiliary and
extra fields in terms of the physical ones. These constraints play
a crucial role in the analysis of the cubic higher spin
interactions. In Sections \bref{sec:aux} and
\bref{sec:gaugefieldsaux} we introduce  bosonic and fermionic
auxiliary variables that enable us to represent   higher spin fields
as expansion coefficients of polynomials in these  variables.
Introducing auxiliary variables is not just a technical tool that
greatly simplifies the whole consideration but also brings to
light such concepts like Howe duality that allows one  to
formulate group-theoretical properties of higher spin fields in a
simple and  manifest fashion. Section \bref{sec: free} also serves
to set our notations and conventions.

In section \bref{sec: FL superalgebra} we review a construction of
Fradkin-Linetsky superalgebra with any number of supersymmetries
$\cN$ giving particular emphasis to the $\cN=2$ case. In Section
\bref{sec:gauging} we describe a gauging procedure that introduces
local symmetry and provides a link to unfolded gauge fields considered in Section
\bref{sec: free}. In Section \bref{sec:supermult} we explicitly
describe the structure of $\cN=2$ higher spin supermultiplets.

Higher spin theories of FV-type are reviewed in Section \bref{sec:reviewFV}.
We formulate all necessary conditions to be satisfied by
the searched-for action in the cubic approximation.
In Section \bref{sec:summary} we formulate the final answer
and list all coefficients in the action both for unreduced and reduced models.
Section \bref{sec:calculations} contains explicit calculations of the coefficients
in the action. Because the total expression for the gauge transformations contains
over a hundred terms  we split them in groups associated with
different gauge supermultiplet parameters and analyze them separately. In Section
\bref{sec: main_bosonic_sym} we explicitly calculate bosonic
gauge invariance for the "hook" fields and this sets a pattern for calculating
the remaining invariance. In Section \bref{sec:remaining_invariance} we
sketch the main steps of how  calculation of the remaining invariance  develops and give the final result
for the coefficient functions collected in Section \bref{sec:summary}.

In Conclusion \bref{sec:conclusion} we shortly discuss our results and
future perspectives.  In Appendix \bref{sec:A} we collect the explicit expressions for the
gauge transformations omitted in the main text.

\section{Free fields}
\label{sec: free}

The isometries of $\ads$ spacetime form $o(4,2)$ algebra and the spectrum of local
excitations of  relativistic  fields is arranged
in terms of labels of irreducible representations of maximally compact subalgebra
$o(2)\oplus o(4) = o(2)\oplus o(3)\oplus o(3)\subset o(4,2)$. The $o(2)$ quantum number physically means the energy $E_0$
while $o(4)$ quantum numbers are spins $(s_1, s_2)$ associated with two $o(3)$ factors
of $o(4)$ subalgebra. For massless gauge fields
quantum numbers are linearly dependent so one may represent the energy
via spin numbers, $E_0 = E_0(s_1,s_2)$, thereby expressing the fact that massless fields
has less degrees of freedom than massive ones \cite{Metsaev:1995re,Brink:2000ag}.
Let $D(s_1, s_2)$ denote a space of states of an $\ads$ massless gauge field. It is
identified with some highest weight unitary irreducible
(infinite-dimensional) representation of $o(4,2)$ algebra.

A (real) number of local degrees of
freedom propagated by massless fields $\#D(s_1, s_2)$ with $s_1>s_2$
has been first calculated in \cite{Metsaev:2004ee} using the light-cone form of higher
spin dynamics
and the answer is given by
\be
\label{degreeB}
\#D(s_1, s_2)= \left\{
\ba{l}
\dps
2s_1+1\,,\;\;\;\; s_1  = n\;,   \;\; s_2 = 0\,, \quad n\in \mathbb{N}\;,
\\
\dps
4s_1+2\,,\; s_1 = n+1/2,\;\;  s_2 = 1/2\;, \quad n\in \mathbb{N}\;,
\\
4s_1+2 \,,\;s_1  = n\;,   \;\; s_2 = k\,, \quad n,k\in \mathbb{N}\;,
\\
\dps
4s_1+2\,,\;s_1  = n+1/2\;,   \;\; s_2 = k+1/2\,, \quad n,k\in \mathbb{N}\;.

\ea
\right.
\ee

\noindent It is important for our future considerations  that
non-symmetric bosonic field ($s_2\neq 0$) have
a number of on-shell degrees of freedom
twice that of totally symmetric bosonic field ($s_2=0$), while fermionic fields have the same degrees
of freedom irrespective of a second spin value.
Massless fields with equal spins $s_1=s_2$
are the so-called doubletons and have no local degrees of freedom
\cite{Gunaydin:1984wc,Gunaydin:1998sw,Ferrara:1998bv}.

In this paper  we use the unfolded formulation of higher spin
dynamics  and  describe massless  gauge fields as
differential $1$-forms taking  values in some irreducible $o(4,2)$ representations (for review see \cite{Bekaert:2005vh} ).
Moreover, we use the  well-known isomorphism
$$o(4,2)\sim su(2,2)$$
and develop a spinor form of the unfolded dynamics in $\ads$ spacetime
\cite{Sezgin:2001zs,Vasiliev:2001wa,Alkalaev:2001mx,Alkalaev:2006hq}.
Fields of the higher spin models under consideration form particular (super)multiplets of
massless bosonic spin-$(s,0)$ fields, fermionic spin-$(s,\half)$ fields,
and massless  spin-($s,1$) "hook" fields. In what follows we explicitly
describe quadratic Lagrangian formulation for these fields giving particular
emphasis to description of "hook" fields.
We start however from describing
$su(2,2)$ spinor form of the gravity thus setting a pattern for higher spin generalizations.

\subsection{Gauge description of $\ads$ spacetime}

$5d$ gravity with the negative cosmological constant
can be formulated in terms of 1-form connection
taking values in the $su(2,2)$ algebra\footnote{Throughout the paper we work within the mostly minus signature
and use notation $\alpha, \beta = 1,..., 4$ for $su(2,2)$ spinor
indices, $i,j = 1,..., \cN$ for $R$-symmetry $u(\cN)$ indices, $\mu,\nu = 0,
..., 4\;$ for  world indices, $a,b = 0, ..., 4$ for tangent
Lorentz $o(4,1)$ vector indices.
}
\be
\label{grav}
\Omega(x) =
dx^{\mu}\:\Omega_{\mu}{}^\alpha{}_\beta(x)\:T_\alpha{}^\beta\;,
\ee
where $T_\alpha{}^\beta$ are basis elements of $su(2,2)$
algebra\footnote{An explicit  realization
of $su(2,2)$ algebra is discussed in  Sections \bref{sec:aux} and  \bref{sec: FL superalgebra}.}
and the connection is traceless, $\Omega_\mu{}^\alpha{}_\alpha = 0$.
As usual, the connection decomposes  into the frame field
and the Lorentz connection. By virtue  of
the compensator mechanism  for gravity theories  this
splitting can be done in a manifestly $su(2,2)$ covariant fashion \cite{Stelle:1979aj,Preitschopf:1997gi}.
For the case at hand  we introduce the  compensator as an
antisymmetric bispinor field
\be
V^{\alpha\beta}(x) = - V^{\beta\alpha}(x)\;,
\ee
normalized so that $V_{\alpha\gamma}V^{\beta\gamma}=\delta_\alpha{}^\beta$
and $V_{\alpha\beta}=\frac{1}{2}\varepsilon_{\alpha\beta\gamma\rho}V^{\gamma\rho}$.
The compensator  field does not carry local degrees of freedom
because  it is an auxiliary field with the transformation law of Stueckelberg type (see \cite{Vasiliev:2001wa} for
more details). The Lorentz subalgebra in $su(2,2)$ is identified with stability transformations of
the compensator.
It follows that the frame field $E^{\alpha\beta}$ and Lorentz spin connection $\omega^{\alpha}{}_\beta$
are defined as \cite{Vasiliev:2001wa}
\be
\label{decomframelor}
\ba{c}
\dps
E^{\alpha\beta} = DV^{\alpha\beta}\equiv dV^{\alpha\beta}
+\Omega^\alpha{}_\gamma V^{\gamma\beta} +\Omega^\beta{}_\gamma
V^{\alpha\gamma}\;,
\qquad
\dps \omega^\alpha{}_\beta = \Omega^\alpha{}_\beta
+\frac{\lambda}{2}\,E^{\alpha\gamma}V_{\gamma\beta}\;,

\ea
\ee
where $\lambda$ is a cosmological parameter, $\lambda^2>0$,
operator $d = dx^\mu \d_\mu$ is the de Rham differential,  and $D$ is the $su(2,2)$ covariant
derivative. Compensator $V^{\alpha\beta}$ is Lorentz-invariant so it can be treated
as a symplectic metric that allows one to raise  and lower spinor indices in a Lorentz
covariant way as
\be
\label{Vsympl}
X^\alpha=V^{\alpha\beta}X_\beta\;,
\qquad
Y_\alpha=Y^\beta V_{\beta\alpha}\;.
\ee
In particular, it follows that $E^{\alpha\beta} V_{\alpha\beta}=0$ and $\omega^{\alpha\beta}V_{\alpha\beta}=0$
which implies that  the frame and Lorentz connection are irreducible Lorentz tensors.

The 2-form curvature $R_\alpha{}^\beta= \half\,  R_{\mu\nu}{}_\alpha{}^\beta dx^\mu\wedge dx^\nu\,$
associated with the connection (\ref{grav}) is given by
\be
R_\alpha{}^\beta= d\:\Omega_\alpha{}^\beta
+\Omega_\alpha{}^\gamma\wedge\Omega_\gamma{}^\beta\;.
\ee
The zero-curvature equation
\be
\label{zerocurv}
R_\alpha{}^\beta(\Omega_0) =0
\ee
locally  describes  metric of $\ads$ spacetime of radius $\lambda^{-1}$. Indeed, decomposing curvature $R_\alpha{}^\beta$ in Lorentz-covariant
components one finds the torsion tensor along with Riemann tensor extended by cosmological term proportional to $\lambda^2$.
Setting these tensors to zero provides a link with
Einstein gravity (see \cite{Bekaert:2005vh} for more details).
The background gravitational fields
will be denoted as $\Omega_0^\alpha{}_\beta=(h^{\alpha\beta},w^{\alpha\beta})$
while the background $su(2,2)$ derivative will be denoted as $D_0$. From
\eqref{zerocurv} it follows  that $D_0$ is nilpotent,  $D_0^2=0$.

\subsection{Totally symmetric massless fields in $\ads$}
\label{sec: totally symmetric}

The metric-like formulation of higher spin dynamics introduces spin-$s$ massless fields
as totally symmetric Lorentz tensors $\phi_{a_1 ... a_s}(x)$ or spin-tensor
$\psi^{\hat \alpha}{}_{a_1 .... a_{s-1/2}}(x)$, where $\hat \alpha$ is a  spinor index.
These (spin-)tensors are gauge fields and transform as
$\delta\phi_{a_1 ... a_s} = {\cD}_{(a_1}\xi_{a_2 ... a_{s})}$ and
$\delta\psi^{\hat \alpha}{}_{a_1 .... a_{s-1/2}} = {\cD}_{(a_1}\xi^{\hat \alpha}{}_{a_2 ... a_{s-3/2})}$,
where ${\cD}$ is a background Lorentz derivative and $\xi_{a_1 ... a_{s-1}}$ and $\xi^{\hat \alpha}{}_{a_1 ... a_{s-5/2}}$
are gauge parameters. Both fields and gauge parameters satisfy certain
algebraic conditions, like trace and gamma-transversality constraints \cite{Fronsdal:1978rb, Fang:1978wz}.

In the framework of the unfolded approach a totally symmetric field of a given spin
is represented  as a differential $1$-form taking values in a definite $o(4,2)$ irreducible
representation \cite{Lopatin:1988hz,Vasiliev:1987tk,Vasiliev:2001wa,Alkalaev:2003qv}.\footnote{Let us
mention  other useful approaches to higher spin dynamics of totally symmetric fields
proposed in Refs.
\cite{Metsaev:1999ui,Segal:2001di,Vasiliev:2001dc,Zinoviev:2001dt,Francia:2002aa,Buchbinder:2002ry,Barnich:2004cr,Barnich:2005ga,Barnich:2006pc,Buchbinder:2007ak}  }
The $su(2,2)$ spinor realization of the unfolded fields
is the following.

\begin{itemize}

\item Spin-$s$ bosonic gauge fields \cite{Sezgin:2001zs,Vasiliev:2001wa}:
\be
\label{b1}
\Omega_\mu{}_{\;\beta_1 ... \beta_{s-1}}^{\;\alpha_1 ... \alpha_{s-1}}
\ee
\item Spin-$s$ fermionic gauge fields \cite{Alkalaev:2001mx,Sezgin:2001zs}:
\be
\label{f1}
\Omega_\mu{}^{\;\alpha_1 ... \alpha_{s-1/2}}_{\;\beta_1 ... \beta_{s-3/2}}
\oplus
\Omega^*_\mu{}^{\;\alpha_1 ... \alpha_{s-3/2}}_{\;\beta_1 ... \beta_{s-1/2}}
\ee

\end{itemize}

\noindent Here symbol $*$ denotes complex conjugation defined by
\be
\label{complex}
(X_\alpha)^* = X^\beta C_{\beta\alpha}\;,
\qquad
(Y^\alpha)^* =C^{\alpha\beta}Y_\beta\;,
\ee
where  $C^{\alpha\beta}=-C^{\beta\alpha}$ and
$C_{\alpha\beta}=-C_{\beta\alpha}$  are some real matrices satisfying
\be
\label{conjmatr}
C_{\alpha\gamma}C^{\beta\gamma}=\delta_\alpha{}^\beta\;.
\ee
We notice  that fermionic fields are described by a pair of mutually conjugated multispinors
while bosonic fields are self-conjugated.
All multispinors are symmetric in upper and lower groups of indices and traceless with respect to
$su(2,2)$ invariant tensor $\delta^\alpha_\beta$. The simplest fields in the list above are Maxwell field $\Omega_\mu$,
Rarita-Schwinger field $\Omega_\mu{}^\alpha$ and its conjugated $\Omega^*_\mu{}_\alpha$,  and the gravity field $\Omega_\mu{}^\alpha{}_\beta$, cf. \eqref{grav}.

Gauge symmetry for the above fields is defined by bosonic $0$-from  parameter $\xi^{\;\alpha_1 ... \alpha_{s-1}}_{\;\beta_1 ... \beta_{s-1}}$ and
mutually conjugated
fermionic $0$-from parameters $\xi ^{\;\alpha_1 ... \alpha_{s-1/2}}_{\;\beta_1 ... \beta_{s-3/2}}$
and $\xi^*{}^{\;\alpha_1 ... \alpha_{s-3/2}}_{\;\beta_1 ... \beta_{s-1/2}}$. The respective transformations
of $1$-from gauge fields are given by
\be
\label{sb1}
\delta \Omega^{\;\alpha_1 ... \alpha_{s-1}}_{\;\beta_1 ... \beta_{s-1}} = D_0 \xi^{\;\alpha_1 ... \alpha_{s-1}}_{\;\beta_1 ... \beta_{s-1}}\;,
\ee
and
\be
\label{sf1}
\delta \Omega^{\;\alpha_1 ... \alpha_{s-1/2}}_{\;\beta_1 ... \beta_{s-3/2}} = D_0 \xi^{\;\alpha_1 ... \alpha_{s-1/2}}_{\;\beta_1 ... \beta_{s-3/2}}
\ee
along with the complex conjugated expression.

The metric-like fields discussed in the beginning of the section are encoded into the
unfolded field \eqref{b1} and \eqref{f1} as their particular components that can be singled out
by imposing particular gauge fixing of the above symmetry. Such a mechanism is similar
to that one used in the gravity theory: the frame field contains a component to be identified
with the metric after gauge fixing local Lorentz symmetry.

\subsection{"Hook" massless fields in $\ads$}
\label{sec: hooks} The first non-trivial example of non-symmetric
fields is given by "hooks" which are bosonic spin-$(s,1)$ massless
fields. They can be described as tensor fields $\phi_{a_1 ...
a_s,\, b_1}(x)$ with two groups of symmetrized Lorentz indices
satisfying Young symmetry condition
\cite{Ouvry:1986dv,Labastida:1986ft}.\footnote{Exhaustive discussion of mixed-symmetry  bosonic gauge fields  both
in Minkowski and AdS spacetimes can be found, \textit{e.g.,} in Refs.
\cite{Siegel:1986zi,Metsaev:1995re,Pashnev:1998ti,Brink:2000ag,
Alkalaev:2003hc,Bekaert:2006ix,Zinoviev:2008ve,Skvortsov:2008vs,
Alkalaev:2008gi,Campoleoni:2008jq,Boulanger:2008up,Buchbinder:2008kw,Alkalaev:2009vm,
Bastianelli:2009eh,Skvortsov:2009zu}.}
The gauge transformations are $\delta \phi_{a_1 ... a_s,\, b_1} =
\cD_{a_1} \xi_{a_2 ... a_{s},\, b_1} +\cD_{b_1}\rho_{a_1 ... a_s}
+ ...$, where the dots denote
appropriate Young symmetrizations needed to adjust symmetry properties of both sides.
Here the gauge parameters $\xi_{a_2 ...
a_{s},\, b_1}$ and $\rho_{a_1 ... a_s}$ are rank-$(s-1,1)$ tensor
and rank-$(s,0)$ tensor, respectively. Both fields and gauge
parameters satisfy certain trace conditions \cite{Labastida:1986ft}.

It is worth noticing that in $5d$ Minkowski spacetime massless spin-$(s_1,1)$ fields
are dual to massless totally symmetric spin-$(s_1,0)$ fields while those with $s_2>1$
do not propagate local degrees of freedom. This fact is in agreement with that
local degrees of freedom of $5d$ Minkowski fields are described by irreducible
tensor representations of little Wigner algebra $o(3)$. In $\ads$ spacetime local degrees
of freedom of massless fields are classified according to $o(3)\oplus o(3)$
so mixed-symmetry massless fields with $s_2>1$ are not dynamically trivial.

A remarkable feature of non-symmetric fields is that they have
different number of gauge symmetries on Minkowski spacetime and $AdS_d$
spacetime \cite{Metsaev:1995re,Brink:2000ag}. Namely, given a mixed-symmetry massless
field in Minkowski spacetime we observe that only a part of
gauge symmetries can be deformed to $AdS_d$ spacetime. In the case under consideration, the symmetry that
survives in $AdS_d$  corresponds to the gauge parameter $\xi_{a_2
... a_{s},\, b_1}$. Lacking one of gauge symmetries on $AdS_d$ results
in a mismatch  between numbers of degrees of freedom propagated by
$\phi_{a_1 ... a_s,\, b_1}(x)$ in Minkowski and $AdS_d$ spacetimes.

In $\ads$ the spinor realization of the unfolded  spin-$(s,1)$ bosonic gauge fields is based on
the following 1-forms \cite{Sezgin:2001zs,Alkalaev:2006hq,Alkalaev:2003qv}:
\be
\label{b2}
\Omega_{\mu}{}^{\;\alpha_1 ... \alpha_{s}}_{\;\beta_1 ... \beta_{s-2}}
\oplus
\Omega^*_{\mu}{}^{\;\alpha_1 ... \alpha_{s-2}}_{\;\beta_1 ... \beta_{s}}
\ee
By analogy with fermionic fields  the "hooks" are complex fields described by a pair of mutually
conjugated multispinors. All multispinors are symmetric in upper and lower groups of indices and
traceless with respect to $su(2,2)$ invariant tensor $\delta^\alpha_\beta$.
Spinor version of gauge symmetry $\xi_{a_2
... a_{s},\, b_1}$ for the $\ads$ "hook" fields  is defined by mutually conjugated $0$-from  parameters
$\xi^{\;\alpha_1 ... \alpha_{s}}_{\;\beta_1 ... \beta_{s-2}}$ and
$\xi^*{}^{\;\alpha_1 ... \alpha_{s-2}}_{\;\beta_1 ... \beta_{s}}$ as
\be
\label{sb2}
\delta \Omega^{\;\alpha_1 ... \alpha_{s}}_{\;\beta_1 ... \beta_{s-2}} = D_0 \xi^{\;\alpha_1 ... \alpha_{s}}_{\;\beta_1 ... \beta_{s-2}}
\ee
along with the complex conjugated expression.

The simplest example of a non-symmetric field, an antisymmetric tensor,
is absent in \eqref{b2}. This happens because  $\ads$ antisymmetric gauge fields are
doubletons  which do not carry local degrees of freedom \cite{Gunaydin:1984wc,Gunaydin:1998sw,Ferrara:1998bv}. Therefore we set
$s>1$ and the first non-trivial example is given by spin-$(2,1)$ field.
Its spinor realization  is given by symmetric bispinor  $\Omega_{\mu}{}^{\alpha\beta}$
along with the  complex conjugated $\Omega^*_{\mu}{}_{\;\alpha\beta}$.

\subsection{Auxiliary spinor variables}
\label{sec:aux}

In practice it is convenient to represent higher spin  fields considered in the previous
sections as expansion coefficients of polynomials with respect to some set of
auxiliary spinor variables. It also brings to light a rich algebraic structure known
as Howe duality that allows one to control   group-theoretical properties
of (spin-)tensor fields in a manifest fashion.

Let us introduce two sorts of auxiliary Grassmann even variables
$a_\alpha$ and $b^{\beta}$, $\alpha,\beta = 1, ...,4$.
It is assumed that $a_\alpha$, $b^{\beta}$
and their derivatives $\dps\dl{a_\alpha}$, $\dps\dl{b^\beta}$ act
in the space $\cP_8$ of polynomials in eight spinor variables
\be
\label{decinAnadB}
F(a,b)  = \sum_{m,n=0}^\infty F^{\alpha_1 ... \alpha_m}_{\beta_1 ... \beta_n}\:a_{\alpha_1}\cdots a_{\alpha_m}\,
b^{\beta_1}\cdots b^{\beta_m}\;,
\ee
where expansion coefficients are multispinors totally symmetric in the upper and lower
groups of indices.

Space $\cP_8$ is a module of $gl(4)$ algebra realized by the following basis elements
\be
\dps G_\alpha{}^\beta  = \half \{a_\alpha,\dl{a_\beta}\} + \half \{b^\beta,\dl{b^\alpha}\}\;,
\ee
that produce  $gl(4)$ commutation relations
via usual commutator. Algebra $gl(4)$ acts homogeneously in $\cP_8$ thereby decomposing it
into finite-dimensional irreducible submodules. The expansion coefficients in
\eqref{decinAnadB} are then identified with $gl(4)$ tensors.

The condition that elements $F(a,b)\in \cP_8$ form an irreducible submodule
under $gl(4)$ transformations is expressed by a set
of the following constraints \cite{Vasiliev:2001wa},
\be
\label{Na}
N_a=a_\alpha\frac{\d}{\d a_\alpha}\;: \quad N_a F(a,b)=m F(a,b)\;,
\ee
\be
\label{Nb}
N_b=b^\alpha\frac{\d}{\d b^\alpha}\;: \quad N_b F(a,b)=n F(a,b)\;,
\ee
where $m$ and $n$ are some integers, and
\be
\label{T-}
T^- =\frac{1}{4}\frac{\d^2 }{ \d a_\alpha \d b^\alpha }\;:
\quad T^- F(a,b)=0\;.
\ee
Then one observes that above operators $N_a, N_b$ and $T^-$
supplemented by
\be
\label{T+}
T^+ = a_\alpha b^\alpha
\ee
form $gl(2)$ algebra. By construction the above $gl(4)$ and $gl(2)$ algebras are
mutually commuting. It is important that  $gl(4)$ invariant conditions \eqref{Na}-\eqref{T-}
are the highest weight (HW) conditions of $gl(2)$ algebra. Indeed,
by an appropriate change of basis one can identify elements
$T^\pm$ with upper- and lower-triangular subalgebras of $gl(2)$ algebra, while
$N_{a, b}$ are its Cartan elements.
Algebra $gl(2)$ can be decomposed in a standard fashion as
$gl(2)=sl(2)\oplus gl(1)$, where the $sl(2)$ part is given by
\be
\label{T}
T^\pm, \quad T^0 =\frac{1}{4} ( N_a +N_b  +4)
\ee
while the following combination
\be
\label{G0}
G^0=N_a-N_b\;
\ee
is identified with  $gl(1)$ basis element. The commutation relations of $sl(2)$
subalgebra are given by
\be
\label{sl2inv}
[ T^0 , T^\pm ] =  \pm  \frac{1}{2}T^\pm \;,
\qquad
[T^- , T^+ ] = T^0\;.
\ee
By definition, element  $G^0$ is central and therefore commutes with any element of
$sl(2)$.

The above consideration also remains valid for
$sl(4)\subset gl(4)$ subalgebra. To this end one notes that condition
\eqref{T-} still defines HW vector of $sl(4)\subset gl(4)$ while
conditions \eqref{Na} and \eqref{Nb} fix some integer weight of $sl(4)$
via
\be
T^0 F(a,b) = \frac{1}{4}(m+n+4)F(a,b)\;,
\ee
along with the following eigenvalue of $gl(1)$
\be
G^0 F(a,b) = (m-n)F(a,b)\;.
\ee
We see that $\cP_8$ is in fact a bimodule over
$gl(4)$ and $gl(2)$ algebras and its structure suggests that  the above two
algebras form Howe dual pair \cite{Howe}.

In addition to commuting auxiliary variables we introduce  auxiliary  Grassmann odd variables $\psi_i$ and $\bpsi^j$
with $i,j  =1,..., \cN$. It enables us to supersymmetrize the above
pure bosonic construction. To this end we introduce a superspace $\cP_{8|2\cN}$
of polynomials
\be
\label{decinAnadBpsi}
F(a,b, \psi, \bpsi)  = \sum_{m,n=0}^\infty \sum_{k,l=0}^\cN F_{\beta_1 ... \beta_m\, |\;j_1\ldots j_l}^{\alpha_1 ... \alpha_n|\;\, i_1\ldots i_k}
a_{\alpha_1} \ldots a_{\alpha_n}\, b^{\beta_1} \dots b^{\beta_m}\psi_{i_1}\cdots \psi_{i_k} \bpsi^{j_1}\cdots \bpsi^{j_l}\;,
\ee
where expansion coefficients are multispinors with two groups of totally symmetric
indices and two groups of totally anti-symmetric indices.
Superspace $\cP_{8|2\cN}$ is a module of $gl(4|\cN)$ superalgebra with the following basis
elements
\be
\ba{c}
\dps G_\alpha{}^\beta  =  \half \{a_\alpha,\dl{a_\beta}\} + \half \{b^\beta,\dl{b^\alpha}\}\;,
\\
\\
\dps
Q_\alpha^i = a_\alpha \bpsi^i+ \dl{b^\alpha}\dl{\psi_i}\;,
\qquad
\bar Q^\alpha_i = b^\alpha\psi_i + \dl{a_\alpha}\dl{\bpsi^i}\;,
\\
\\

\dps
U_i{}^j = \half \big[\psi_i, \dl{\psi_j}\big]+\half \big[\bpsi^j,\dl{\bpsi^i}\big]\;,
\ea
\ee
and $\cP_{8|2\cN}$ decomposes into $gl(4|\cN)$ invariant  submodules.

Introducing Grassmann odd variables
enables one to extend the above bosonic realization of $gl(2)$ algebra. The respective basis elements of $sl(2)$
are given by \cite{Alkalaev:2002rq}
\be
\label{fermionsl(2)}
P^+ = T^+-\psi_i\bpsi^i\,,
\quad
P^- =T^- +\frac{1}{4}\frac{\d^2 }{ \d \bpsi^i \d \psi_i } \,,
\quad
P^0 =T^0+\frac{1}{4}(N_{\psi} +N_{\bpsi}-\cN)\,\,,
\ee
and $gl(1)$ basis element is
\be
\label{Z0}
Z^0 = G^0 + N_{\psi} - N_{\bpsi}\;,
\ee
where
\be
\label{Npsi}
\dps N_{\psi}=\psi_i\frac{\d}{\d \psi_i}\,,
\qquad
N_{\bpsi}=\bpsi^i\frac{\d}{\d \bpsi^i}\,.
\ee
The respective $sl(2)$ commutation relations are
\be
[ P^0 , P^\pm] =  \pm \frac{1}{2} P^\pm \;,
\qquad
[P^- , P^+ ] =  P^0\;.
\ee
By construction, the above $gl(2)$ algebra  and $gl(4|\cN)$ superalgebra
are mutually commuting and form  Howe dual pair.
It makes possible to study $gl(4|\cN)$ irreducible submodules in $\cP_{8|2\cN}$
via imposing the following $sl(2)$ HW condition
\be
\label{HW}
P^- F(a,b, \psi, \bpsi) = 0 \;,
\ee
along with some fixed eigenvalues of  $sl(2)$ Cartan element $P^0$ and
$gl(1)$ element $Z^0$. It is worth noting that the present construction
describes only  particular class of $gl(4|\cN)$ irreducible representations.

Up to now we considered complex $gl(4|\cN)$ superalgebra.
However, we are interested in  $su(2,2|\cN)$ superalgebra that is defined as
an appropriate real form of $sl(4|\cN)\subset gl(4|\cN)$. The respective reality condition are
given by
\be
\label{auxinv}
({a}_\alpha)^*   = b^\beta C_{\beta
\alpha}\,,\quad  ({b}^\alpha)^*  = C^{\alpha\beta} a_\beta
\,,
\qquad
(\psi_i)^* = \bpsi^i\;,\quad
(\bpsi^i)^* = \psi_i\;,
\ee
where conjugation matrices are defined by \eqref{conjmatr}.
Then it follows that  $a_\alpha$ and $b^\beta$ are in the fundamental
and the conjugated fundamental representations of $su(2,2)$ while
$\psi_i$ and $\bpsi^i$ are in the fundamental
and the conjugated fundamental representations of $u(\cN)$.

In Section \bref{sec:oscillator} we discuss a star-product realization of the above
construction.
Finally, we note that the
Howe dual pair $gl(4|\cN)$ and $gl(2)$ coincides with that one discussed
in \cite{Alkalaev:2008gi} within the BRST framework.

\subsection{Gauge fields as polynomials in auxiliary variables}
\label{sec:gaugefieldsaux}

The unfolded gauge fields discussed in Sections \bref{sec: totally symmetric} and
\bref{sec: hooks} can be collectively represented as a pair of mutually
conjugated multispinors
\be
\label{sreda}
\Omega_\mu{}^{\alpha_1...\, \alpha_{s_1+s_2-1}}_{\;\beta_1 ... \, \beta_{s_1-s_2-1}}
\oplus \Omega_\mu^*{}_{\;\alpha_1...\, \alpha_{s_1+s_2-1}}^{\,\beta_1 ... \, \beta_{s_1-s_2-1}}\;,
\ee
provided that $s_1= s$ and $s_2 = 0,\half, 1$.
Using spinor auxiliary variables introduced in the previous section we
define the above massless gauge fields
as follows
\be
\label{gen}
\Omega(a, b|x) =
dx^\mu\,\Omega_\mu{}^{\alpha_1...\, \alpha_{s_1+s_2-1}}_{\beta_1 ... \, \beta_{s_1-s_2-1}}(x)
\:
a_{\alpha_1}...\, a_{\alpha_{s_1+s_2-1}}b^{\beta_1} ...\, b^{\beta_{s_1-s_2-1}}
\;
\ee
along with the complex conjugated $\Omega^*(a, b|x)$.
The associated linearized higher
spin curvature is a $2$-from $R_1 = \dps \half dx^\mu\wedge dx^\nu R_1{}_{\,\mu\nu}(a,b|x)$ given by
\be
\label{curv}
R_1 = D_0 \Omega \equiv  d\:\Omega +
\Omega_0{}^\alpha{}_\beta (b^\beta\frac{\d}{\d b^\alpha}-a_\alpha\frac{\d}{\d
a_\beta})\wedge \Omega\;,
\ee
where $\Omega_0{}^\alpha{}_\beta$ is the background 1-form connection satisfying the
zero-curvature condition (\ref{zerocurv}), and the background covariant
derivative is given by
\be
D_0=d +\Omega_0{}^\alpha{}_\beta (b^\beta\frac{\d}{\d
b^\alpha}-a_\alpha\frac{\d}{\d a_\beta})\;.
\ee
Subscript $1$ indicates that  curvature \eqref{curv}
is a  linearized part of some full non-Abelian curvature introduced in Section \bref{sec:gauging}.
The gauge transformations are
\be
\label{hstr}
\delta\Omega=D_0\xi\;,
\ee
where a gauge parameter is a 0-form $\xi = \xi(a,b|x)$.
As a corollary of $D_0^2=0$ it follows that
\be
\label{curvinar}
\delta R_1(a,b|x)=0\;,
\ee
while the respective Bianchi
identities read as
\be
D_0R_1(a,b|x)=0\;.
\ee
Using  $gl(2)$ basis elements (\ref{T}) one easily formulates algebraic conditions
on $\Omega(a,b|x)$ that single out an irreducible field of a given spin as
the respective $gl(2)$ HW condition
\be
\label{HWgauge}
T^- \Omega(a,b|x)=0\;,
\ee
along with particular eigenvalues of Cartan elements
\be
\label{IrreducibleSpin}
\ba{l}
\dps
N_a \Omega(a,b|x) = (s_1+s_2-1)\Omega(a,b|x)\;,
\\
\\
\dps
N_b \Omega(a,b|x) = (s_1-s_2-1)\Omega(a,b|x)\;.
\ea
\ee
The last two conditions can be equivalently rewritten as
\be
\ba{c}
\dps
T^0\Omega(a,b|x) = \frac{1}{2}(s_1+1)\Omega(a,b|x)\;,
\\
\\
\dps
G^0\Omega(a,b|x)  = 2s_2\Omega(a,b|x)\;.
\ea
\ee
It is obvious that the associated curvatures satisfy the same algebraic constraints.

 In the subsequent analysis we use the following set of differential operators
in auxiliary spinor variables   \cite{Vasiliev:2001wa}
\be
\label{S}
S^- = a_\alpha \frac{\d}{\d b^\beta}V^{\alpha\beta}\:,\qquad
S^+ = b^\alpha \frac{\d}{\d a_\beta} V_{\alpha\beta}\:,\qquad
\dps S^0 = N_b - N_a\;.
\ee
They explicitly involve the compensator field and form $sl(2)$ algebra
\be
\label{sl2invS}
[ S^0 , S^\pm ] =  \pm  \frac{1}{2}S^\pm \;,
\qquad
[S^- , S^+ ] = S^0\;.
\ee
Note that the above set of $sl(2)$ elements commute with other $sl(2)$ elements
introduced earlier in Section \bref{sec:aux}. It is worth noting that $sl(2)$ algebra
\eqref{sl2invS} can be interpreted as Howe dual algebra for the Lorentz subalgebra of
$su(2,2)$. We hope to consider this issue in a more detail  elsewhere.

Irreducible $su(2,2)$ gauge fields  can be further decomposed with
respect to Lorentz subalgebra. The resulting Lorentz fields are given by the following
collection of differential $1$-forms
\be
\label{su_lor_fields}
\omega^t_\mu{}^{\;\alpha_1... \, \alpha_{s_1+s_2+t-1},\,\beta_1 ... \, \beta_{s_1-s_2-t-1}}(x)
\;,
\qquad
0\leq t\leq s_1-s_2-1\;,
\ee
that satisfy
the Young symmetry condition and the $V_{\alpha\beta}$-transversality condition.
Recall that compensator $V^{\alpha\beta}$ can be used to raise and lower indices in the Lorentz-invariant manner,
see \eqref{Vsympl}.
Fields \eqref{su_lor_fields} can be described as expansion coefficients of
\be
\label{omegas}
\omega^{t}(a,b|x) = d x^\mu\omega_\mu^{t}(a,b|\,x)\;.
\ee
Irreducibility
conditions imposed on Lorentz-covariant tensors have the form of two $gl(2)$ HW conditions
\be
\label{irr}
S^-\omega^{t} = 0\;, \qquad T^-\omega^{t} = 0\;.
\ee
The first condition is in fact the Young symmetry condition while the second one
tells us that Lorentz tensors are transversal to  compensator $V^{\alpha\beta}$.
The last condition expresses  the fact that we describe Lorentz irreps
in a manifestly $su(2,2)$ covariant manner. Indeed, operators \eqref{S} enables one
to write down a decomposition of an irreducible $su(2,2)$ gauge field as
\be
\label{expan}
\Omega(a,b| x) = \sum_{t=0}^{\sm} (S^+)^t \:\omega^{\,t}(a,b|x).
\ee
Since  $sl(2)$ algebras
\eqref{sl2inv} and \eqref{sl2invS} mutually commute one concludes that the second
HW condition in \eqref{irr} on $\omega^{\,t}(a,b|x)$ follows from
HW condition \eqref{HWgauge} on $\Omega(a,b|x)$.

The background covariant derivative can be cast into explicit Lorentz-covariant  form
as $D_0 = \cD_0 +\sigma_- + \lambda\sigma_0 + \lambda^2\sigma_+ $,
where $\cD_0$ stands for  Lorentz derivative constructed with respect to
background Lorentz connection $w^{\alpha\beta}$, while $\sigma$-operators
satisfy the relations
\be
\label{sigmas}
(\sigma_\pm)^2=0\;,
\quad
\{\sigma_0,\sigma_\pm \} = 0\;,
\quad
{\cal D}^2+\lambda^2\, \{\sigma_-,\sigma_+\}+\lambda^2\,(\sigma_0)^2  =0\;,
\ee
that follow from $D_0^2 = 0$. The explicit expressions for $\sigma$-operators are given  in \cite{Alkalaev:2006hq}.
It is worth noting that non-trivial $\sigma_0$ appears not only for fermionic totally symmetric
fields but also for  bosonic and fermionic mixed-symmetry fields.

Lorentz-covariant fields $\omega^{t}$ at different values of parameter $t$ play
different dynamical roles. One distinguishes between physical, auxiliary, and extra fields.

\begin{itemize}

\item For integer spin-$(s,0)$ system:  fields with $t=0$ are called physical, fields with $t=1$
are auxiliary ones, fields with $t>1$ are called extra fields.

\item For half-integer spin-$(s,\frac{1}{2})$ system: fields with $t=0$ are physical ones, fields
with $t>0$ are extra fields. The absence of auxiliary fields is a manifestation
of the first-order form of the fermionic field equations.

\item For integer spin-$(s,1)$ system: fields with $t=0$ are physical and auxiliary ones, fields
with $t>0$ are extra fields. Physical field is identified with ${\rm Re}\, \omega^{t}$, an auxiliary
field is identified with ${\rm Im}\,\omega^{t}$. In particular, it allows one to cast the dynamical
equations of non-symmetric fields into the first-order form.
The analogous decomposition into pure real and
imaginary parts duplicates the number of (real) extra fields. For more details see \cite{Alkalaev:2006hq}.

\end{itemize}

The unfolded dynamical higher spin equations of motion can be represented
as a system of variational equations and certain constraints. Variational equations
involve just physical and auxiliary fields,  and auxiliary field  is expressed
via first derivative of the physical field, while the constraints
express all extra fields via derivatives of the physical field.
The next two sections discuss the action functional and the appropriate constraints.

\subsection{Higher spin  action functionals}

One of basic advantages of using the unfolded formulation is that  quadratic action
functionals for higher spin fields can be represented in a manifestly gauge-invariant fashion. The
actions
have the form of a bilinear combination of  linearized curvatures so
the gauge invariance of the action is a direct consequence of \eqref{curvinar}.

The $\ads$  action functional
involves HS fields  described  as polynomials in  two sets of auxiliary variables
$X_1=(a_1{}_\alpha, \beta_1^\beta)$ and $X_2=(a_2{}_\alpha, \beta_2^\beta)$. The action
functional is built then in the following schematic form
\be
\label{actschem}
S = \int_{\cM^5}\hat H\Big(E, V, \frac{\d}{\d X_1}\,, \frac{\d}{\d X_2}\Big) \wedge \; R(X_1)\wedge \;R(X_2)
\Big|_{X_1=X_2=0}\;,
\ee
where $\hat H$ is a polynomial in the compensator and  auxiliary variable derivatives
acting on a tensor product of two field strengths $R(X_{})$. Also, since the integrand  is
$5$-form it follows that $\hat H$ is a 1-form proportional to the  frame
field $E_\mu^{\alpha\beta}$. Expansion coefficients of
$\hat H$  with respect to derivatives in auxiliary variables are some $su(2,2)$ covariant
tensors built of $V^{\alpha\beta}$ and $\delta^\alpha_\beta$ and their combinations
parameterize various types of index contractions
between curvatures. Any such action is manifestly $su(2,2)$
invariant and automatically gauge-invariant with respect to the gauge transformations (\ref{hstr}).

Generally, actions of the type \eqref{actschem} do not describe propagation of a correct
number of on-shell degrees of freedom because of redundant dynamical modes associated with the
extra fields. In order to eliminate their contribution one should fix the operator $\hat  H$ in an
appropriate form by virtue of the extra field  decoupling condition. It requires that
the variation of the quadratic action with respect to  extra fields is identically zero,
\be
\label{exdc}
\frac{\delta { S}_2^{}}{\delta\omega^{ex}} \equiv 0\;.
\ee
Extra fields maintain an explicit gauge invariance of the action functional but the above
condition constrains them to fall out of the quadratic action.
Having decoupled extra fields, the action can be cast into a minimal form with just two
fields, physical and auxiliary ones, but then the residual gauge invariance is implicit.
Nonetheless, for both versions of the action, minimal form with two fields and
non-minimal with added extra fields,  the respective  free field equations of motion always
have manifestly gauge-invariant form, \textit{i.e.}, they are represented as linear combinations
of linearized higher spin curvatures.

The action for spin-$(s_1,s_2)$ massless gauge field is searched in the
following form \cite{Vasiliev:2001wa,Alkalaev:2001mx,Alkalaev:2006hq}
\be
\label{act}
S^{(s_1, s_2)}_2 =\int_{{\cal M}^5} {\hat H}\wedge  R(a_1,b_1)\wedge
{R}^*(a_2,b_2)|_{a_i=b_i=0}\;,
\qquad
s_2 = 0,\half, 1\;,
\ee
where $R$ and $R^*$ are  mutually conjugated  linearized spin-$(s_1,s_2)$ curvatures (\ref{curv})
and $\hat H$ is the following 1-form differential operator
\be
\label{H}
\ba{ccc}
\dps \hat  H=
\Big(\alpha(p,q) E_{\alpha\beta} \frac{\d^2}{\d a_{1\alpha} \d
a_{2\beta}}{b}_{12}
+\beta(p,q) E^{\alpha\beta} \frac{\d^2}{\d b_1^\alpha \d
b_2^\beta}{a}_{12} & &
\\
\\
\dps
& & \dps \hspace{-9cm}+\gamma(p,q) E_\alpha{}^\beta \frac{\d^2}{\d a_{2\alpha} \d b_1^\beta}{c}_{12}
+\zeta(p,q) E_\alpha{}^\beta \frac{\d^2}{\d a_{1\alpha} \d b_2^\beta}{c}_{21}
\Big)\,({c}_{12})^{2s_2}\;.
\ea
\ee
Here $E^{\alpha\beta}$ is the frame field \eqref{decomframelor}. For quadratic
action under consideration the frame field is taken to be background
\be
E^{\alpha\beta} = h^{\alpha\beta}\;,
\ee
so dynamical fields are contained
in the linearized curvatures only.
The coefficients $\alpha,\;\beta,\;\gamma$ and $\zeta$  are functions of operators
\be
p={a}_{12}{b}_{12}\;,\qquad q={c}_{12}{c}_{21}\;,
\ee
where
\be
\label{abg}
\ba{cc}
\dps{a}_{12} = V_{\alpha\beta}\frac{\d^2}{\d a_{1\alpha} \d
a_{2\beta}}\;,&
\qquad\dps {b}_{12} = V^{\alpha\beta}\frac{\d^2}{\d b_1^\alpha \d
b_2^\beta}\;,
\\
\\
\dps{c}_{12} = \frac{\d^2}{\d a_{1\alpha} \d b_2^\alpha}\;,&
\dps{c}_{21} = \frac{\d^2}{\d a_{2\alpha} \d b_1^\alpha}\;.
\\
\\
\ea
\ee
These functions are responsible for various types of index  contractions
between the background frame field, compensator and curvatures.

Below we list solutions of the extra field decoupling condition  for
totally symmetric bosonic and  fermionic fields, and for bosonic "hook" fields.
Note that quadratic actions are defined modulo total derivative contributions.

\begin{itemize}

\item Spin-($s,0$) bosons:

\be
\label{K1}
\ba{c}
\dps
\alpha(p,q) =   2 \int_0^1 d\tau \big(1+q\frac{\d}{\d q} \big)\rho(\tau p+q)\;,
\\
\\
\zeta(p,q) + \gamma(p,q) = 0\;,
\qquad
\beta(p,q) = 0\;,
\qquad
\gamma(p,q) = \rho(p+q)\;.

\ea
\ee

\item Spin-$(s,\frac{1}{2})$ fermions:

\be
\label{K2}
\ba{c}
\dps
\alpha(p,q) =   - \int_0^1 d\tau \frac{\d}{\d p}\rho(p\tau + q)\;,
\\
\\
\dps
\gamma(p,q) = 0\;,
\qquad
\beta(p,q) = 0\;,
\qquad
q\zeta(p,q) = \rho(p+q)\;.
\ea
\ee

\item Spin-$(s,1)$ bosons:

\be
\label{K3}
\ba{c}
\dps
\dps \alpha(p,q) = -\int_0^1 d\tau \frac{\d}{\d p}\rho(p\tau+q)\;,
\\
\\
\dps
\gamma(p,q) = 0\;,
\qquad
\beta(p,q) = 0\;,
\qquad
q\zeta(p,q) = \rho(p+q)\;.
\ea
\ee

\end{itemize}
We see that the quadratic action for a given spin is fixed unambiguously up
to overall factors parameterized by polynomials  $\rho(p+q)$ of fixed order, $\rho(p+q) = \rho_0 (p+q)^{s_1-2}$ for
$s_2=0$ fields and $\rho(p+q) = \rho_0(p+q)^{s_1-s_2-1}$
for $s_2\neq 0$ fields, $\rho_0$ are arbitrary constants. Constants $\rho_0$ cannot be fixed from the
free field analysis and represent the leftover ambiguity in
the coefficients.
On the other hand, requiring gauge invariance in the cubic theory fixes
$\rho(p+q)$ unambiguously, see Section \bref{sec:reviewFV}.

\subsection{Generalized Weyl tensors and constraints}
\label{sec:weyl}

As discussed in the previous section in order to have a manifest higher spin gauge invariance,
the quadratic action is always written down with the extra fields, at least formally. It turns out that
on the interaction level variation of the action with respect to extra fields cannot be consistently required
to vanish identically. It follows that proper constraints should be imposed expressing extra fields in terms
of physical fields  thereby preserving a correct number of gauge symmetries and physical degrees of freedom.

We assume that constraints for extra fields should have the following form \cite{Vasiliev:2001wa,Alkalaev:2002rq,Alkalaev:2006hq}
\be
\label{constr}
\hat \Upsilon_2^+ \wedge r^t=0\;,\quad 0\leq t<s_1-s_2-1\;,
\quad s_2 = 0,\half,1\;,
\ee
where $\hat \Upsilon_2^+ $ is some 2-form operator increasing grading $t$ by one. It satisfies the condition
\be
\label{prop}
\sigma_+\wedge \hat \Upsilon_2^+ =0
\ee
that guarantees that the number of independent algebraic relations imposed on the curvature $r^t$ coincides with the number of
components of extra fields $\omega^{t>0}$ modulo pure gauge components of the form
$\delta\omega^{t+1} = \sigma_-\xi^{t+2}$. One can show that the operator
$\hat \Upsilon_2^+$ is uniquely fixed in the form
\be
\hat \Upsilon_2^+ = {\sigma}_0\wedge {\sigma}_+\;.
\ee
Constraints \eqref{constr} are described by 4-form which in $d=5$ dimensions is dual to 1-form
so it follows that the number of
equations in (\ref{constr}) coincides with the number of components
of $\omega^{\,t+1}$. Therefore, field $\omega^{t+1}$ can be expressed via derivatives of the
field $\omega^{t}$ for any $t>0$. Finally, one can obtain fields $\omega^{t}$ expressed in terms of derivatives of the field
$\omega^0$ with an order of highest derivatives equal to $t$.
The schematic form of the corresponding expressions is
\be
\label{coco}
\omega^t \sim  \Big(\frac{\d}{\lambda \d x}\Big)^t\,  \omega^0\;.
\ee
On the non-linear level such expression for extra fields provide a useful parameterization
of higher derivatives in the higher spin interaction terms.

Next we cite  the proposition
known in the literature as the first on-mass-shell theorem, see,
\textit{e.g.}, \cite{Vasiliev:1986td,Lopatin:1988hz,Vasiliev:1987tk,Sezgin:2001zs}.

\begin{prop}
\label{fronprop}
Variational equations of motion for spin-$(s_1,0)$  and spin-$(s_1,1/2)$ fields
supplemented with the constraints for extra fields can be equivalently
rewritten as
\be
\label{COMST1}
R^{\alpha_1 ... \alpha_{s_1-1}}_{\beta_1 ... \beta_{s_1-1}} =
H_2{}_{\delta\rho} C_0^{\alpha_1 ... \alpha_{s_1-1} \gamma_{1} ... \gamma_{s_1 - 1}\delta\rho}
V_{\gamma_{1} \beta_{1}} \cdots V_{\gamma_{s_1-1} \beta_{s_1-1}}\;,
\ee
and
\be
\label{COMST2}
R^{\;\alpha_1 ... \alpha_{s_1-1/2}}_{\;\beta_1 ... \beta_{s_1-3/2}} =
H_2{}_{\,\delta\rho} C_{1/2}^{\alpha_1 ... \alpha_{s_1-1/2}
\gamma_{1} ... \gamma_{s_1 - 3/2}\delta\rho}
V_{\gamma_{1}\beta_{1}} \cdots V_{\gamma_{s_1 - 3/2} \beta_{s_1 - 3/2}}\;,
\ee
plus analogous expression for complex conjugated curvatures.
Here $H_2{}_{\,\delta\rho} = h_\delta{}^\gamma \wedge h_{\gamma\rho}$.
Totally symmetric multispinor $C_0^{\alpha_1 ... \alpha_{2s_1}}$ is a generalized bosonic Weyl tensor, and
totally symmetric multispinors $C_{1/2}^{\alpha_1... \alpha_{2s_1}}$ and its complex
conjugated constitute generalized fermionic Weyl tensor.
\end{prop}

\vspace{2mm}

In particular, the above proposition tells us that  all Lorentz-covariant curvatures except for that
with $t=s_1-1$ for bosons and $t=s_1-3/2$ for fermions can be set to zero
on-shell provided appropriated constraints are imposed.
The proposition generalizes the well-known construction of  Weyl tensor in gravity.

Now we formulate and prove the analogous proposition for $\ads$ mixed-symmetry fields
of particular integer spin $(s_1,1)$. Actually, higher spin field equations of the form $R = H_2 C$
are known to describe  $\ads$ "hook" field dynamics \cite{Sezgin:2001zs}. Here
we prove that these equations do arise as variational equations supplemented
by some constraints thus guaranteeing  the proposed action functional for "hook"
fields \eqref{act} correctly describes physical degrees of freedom.

\begin{prop}
\label{hookprop}
Variational equations of motion for spin-$(s_1,1)$ fields supplemented with the constraints can be equivalently
rewritten as
\be
\ba{l}
R^{\;\alpha_1 ... \alpha_{s_1}}_{\beta_1 ... \beta_{{s_1}-2}} =
H_2{}_{\delta\rho} C_1^{\alpha_1 ... \alpha_{s_1} \gamma_{1} ... \gamma_{{s_1} - 2}\delta\rho}
V_{\gamma_{1}\beta_{1}} \cdots V_{\gamma_{{s_1} - 2} \beta_{{s_1} - 2}}
\\
\\
\dps
R^*{}_{\beta_1 ... \beta_{s_1}}^{\;\alpha_1 ... \alpha_{{s_1}-2}} =
H_2{}^{\,\delta\rho} C^*_1{}_{\beta_1 ... \beta_{s_1} \gamma_{1} ... \gamma_{{s_1}-2} \delta\rho}
V^{\gamma_{1} \alpha_{1}} \cdots V^{\gamma_{{s_1}-2} \alpha_{{s_1}-2}}\;.
\ea
\ee
Here $H_2{}_{\,\delta\rho} = h_\delta{}^\gamma \wedge h_{\gamma\rho}$. Totally symmetric multispinors
$C_1^{\alpha_1... \alpha_{2s}}$ and $C^*_1{}_{\beta_1 ... \beta_{2s}}$ are
mutually conjugated and constitute generalized Weyl tensor for "hook" fields.
\end{prop}

\begin{proof} The main idea behind the proof is to observe that both  the variational
equations of motion for the physical and auxiliary fields and constraints for extra fields
can be visualized  as a system of  linear equations imposed on curvature components.
The kernel of the linear system  should be identified with generalized Weyl
tensors so finding it is in fact the content of the above Proposition.

More precisely, Lorentz-covariant curvatures can be cast into the following form
\be
r_{\mu\nu}{}^{\alpha_1...\, \alpha_{s+t},\, \beta_1 ...\, \beta_{s-t-2}}
\Rightarrow
r^{(\delta\rho)|}{}^{\alpha_1...\, \alpha_{s+t},\, \beta_1 ...\, \beta_{s-t-2}}\;,
\ee
where antisymmetric 2-form indices were converted to symmetric spinor indices
by virtue of a $2$-from composed of the background frame field, $H_2{}_{\,(\delta\rho)} = h_\delta{}^\gamma \wedge h_{\gamma\rho}$. The tensor product
$(\delta\rho)\otimes (\alpha_1...\, \alpha_{s+t},\, \beta_1 ...\, \beta_{s-t-2})$
contains a set of irreducible Lorentz-covariant multispinor components $r^{\alpha_1 ... \alpha_k,\,\beta_1 ... \beta_l}$
for some definite integers $k,l$. Field equations
and constraints impose various linear relations on these components.

As a first step we consider the curvature $t=0$ and analyze which of its components
do not vanish on the equations of motions. The
equations of motion have the form $\hat\cE \wedge r^{\,t=0} = 0$, where $\hat\cE$ is
a 2-form  operator satisfying the conditions $[S^-,\hat\cE]=0$ and $[T^-,\hat\cE]=0$,
{\it i.e.} when acting on the Lorentz-covariant curvature it preserves its Young symmetry
and $V^{\alpha\beta}$-transversality
properties. Operator $\hat\cE$ is proportional to background frame 2-form
$H^{\alpha\beta}$ and is a differential operator
in auxiliary spinor variables.  The exact expression for $\hat \cE$ can be found in \cite{Alkalaev:2006hq}.

The explicit analysis of the component form of equations of motion  is straightforward but
technically involved to be given here in all detail. However, since we work with linear equations it is
possible
to estimate the lower bound of the respective kernel dimension just by comparing the numbers of variables and equations by
$\#({\rm kernel}) = \#({\rm variables})-\#({\rm equations})$. The explicit analysis confirms that
a rank of the linear system is maximal and the above formula is exact.

Denoting $t=0$ curvature component $r^{\alpha_1 ... \alpha_k,\,\beta_1 ... \beta_l}$ satisfying the Young symmetry and $V^{\alpha\beta}$-transversality
conditions as a pair $(k,l)$ we find that the following multispinor components of $t=0$ curvature  remain
non-zero on-shell: $(s+2,s-2)$, $(s,s-4)$ $(s,s)$, $(s-2,s-2)$.

Consider the Bianchi identities for $t=0$ curvature,
$\cD r^{\,t=0} + \lambda \,\sigma_0 r^{\,t=0} + \sigma_{-}r^{\,t=1}=0$, where $\sigma$-operators satisfy
\eqref{sigmas}. Projecting these Bianchi identities on  components $(s,s)$ and  $(s-2,s-2)$
gives rise to conditions  $\lambda\, r^{\alpha_1 ... \alpha_s,\,\beta_1 ... \beta_s}=0$
and $\lambda\, r^{\alpha_1 ... \alpha_{s-2},\,\beta_1 ... \beta_{s-2}}=0$. Note that
these components originate from the term with $\sigma_0$ operator.

Then we consider   constraints \eqref{constr} at $t=0$. They can be equivalently represented
as $1$-form taking values in $(s-1,s-3)$ multispinor corresponding to extra field with $t=1$.
It implies that by virtue of this constraint the $t=1$ extra field can be completely expressed as
the first derivative of $t=0$ field. Again, considering the constraint as a system of linear equations
on $t=0$ curvature one finds that its $(s+2,s-2)$, $(s,s-4)$ components vanish.

To summarize, we proved that equations of motion along with the first of constraints can be equivalently
rewritten as $r^{\,t=0}=0$. The rest of the proof is straightforward and reduces to the observation
that curvatures $r^{\,t>0}$ satisfy the cohomological equation $\sigma_- r^{\,t>0}=0$
as it follows form the respective Bianchi identities. Modulo exact contributions the general solution
is
\be
\ba{l}
\dps
r^{\alpha_1...\, \alpha_{s+t},\, \beta_1 ...\, \beta_{s-t-2}} = 0\;,
\qquad\qquad\;\;
0< t < s-2\;,
\\
\\
\dps
r^{\alpha_1...\, \alpha_{2s-2}} = H_2{}_{\delta\rho} C_1^{\alpha_1 ... \alpha_{2s - 2}\delta\rho}\;,
\qquad\quad
 t = s-2\;,
\ea
\ee
where $H_2{}_{\,\delta\rho} = h_\delta{}^\gamma \wedge h_{\gamma\rho}$ and totally symmetric multispinor
$C_1^{\alpha_1 ... \alpha_{2s}}$ should be identified with the
generalized Weyl tensor. The analogous expression is valid  for complex conjugated curvatures.
We see that the above expressions can be equivalently cast into the form of the Proposition \bref{hookprop}.

\end{proof}

\vspace{-3mm}

Using auxiliary variables expression (\ref{COMST1}), \eqref{COMST2}, (\ref{hookprop}) can be uniformly cast into the following
form
\be
\label{onmassH21}
R^{s_1,s_2}(a,b|x)=H_{2\,\alpha}{}^\beta \frac{\d^2}{\d
a_\alpha \d b^\beta}{\rm Res}_\mu (\mu^{2s_2}\,C^{s_1,s_2}(\mu
a+\mu^{-1}b|x))\,,
\ee

\be
\label{onmassH22}
R^*{}^{s_1,s_2}(a,b|x)=H_{2\,\alpha}{}^\beta \frac{\d^2}{\d
a_\alpha \d b^\beta}{\rm Res}_\mu (\mu^{-2s_2}\, C^*{}^{s_1,s_2}(\mu
a+\mu^{-1}b|x))\,,
\ee

\noindent where $s_2 = 0,\half,1$, and
$H_{2\,\alpha\beta}=h_{\alpha\gamma}\wedge h^\gamma{}_\beta$ and ${\rm Res}_\mu $
singles out the $\mu$-independent part of Laurent series in $\mu$.
A function of one spinor variable
\be
C(\mu a + \mu^{-1}
b) = \sum_{k,\,l} \frac{\mu^{k-l}}{k!\, l!} C^{\alpha_1 \ldots
\alpha_{k}\beta_1 \ldots \beta_{l}} a_{\alpha_1} \ldots
a_{\alpha_k} b_{\beta_1} \ldots b_{\beta_l}
\ee
has totally symmetric coefficients $C^{\alpha_1 \ldots \alpha_{k}\beta_1
\ldots \beta_{l}}$ and $b_\beta=b^\gamma V_{\gamma\beta}$.

We  observe  that generalized Weyl tensor for bosonic non-symmetric spin-$(s_1, s_2)$ field
is given by  a pair of mutually conjugated
generalized Weyl tensors for totally symmetric spin-$s_1$ field. In particular, it
implies that the number of physical degrees of freedom is twice that of
symmetric spin-$s_1$ field, cf.  \eqref{degreeB}. In the flat limit $\lambda=0$ the above
mixed-symmetry field decomposes into two independent totally symmetric spin-$s_1$ fields.
Indeed, there are no mixed-symmetry fields on Minkowski spacetime since
the respective little Wigner algebra $o(3)$ has just totally symmetric representations.
It conforms
the Brink-Metsaev-Vasiliev conjecture \cite{Brink:2000ag} asserting that an irreducible massless field in $AdS_d$
decomposes in the flat limit into a collection of irreducible massless fields in Minkowski
spacetime.\footnote{The conjecture was originally put forward in the
group-theoretical terms while its field-theoretical justification based on
unfolded formalism has been proposed in \cite{Boulanger:2008up} for $AdS_d$
mixed-symmetry fields of general shape. There, however, the proof could only be provided in full rigor
for fields up to four rows, due to technicalities in the manipulation of
so-called cell-operators. The proof of the conjecture in the general case was
given in \cite{Alkalaev:2009vm} where  BRST extension of the unfolding formalism was
used, that dispensed the authors of \cite{Alkalaev:2009vm} with an explicit manipulation of
cell-operators.} In the case of $\ads$ spacetime a set of Minkowski fields
drastically reduces so that  a non-symmetric  bosonic field  decomposes into
a pair of equal spin totally symmetric Minkowski fields \cite{Metsaev:2004ee}.

To conclude this section one should note that the naive flat limit $\lambda=0$ of the unfolded quadratic
action \eqref{act} for $\ads$ massless spin-$(s,1)$ fields is inconsistent in the sense
a half of PDoF is lost \cite{Brink:2000ag}.  However, such a type of inconsistency
may be ignored on the non-linear level since
the higher spin interaction terms contain a factor of $\lambda^{-1}$  so
the naive flat limit in the $\ads$ interacting theory is singular. This drawback could be cured
within the Stueckelberg-like approach developed for mixed-symmetry fields in
\cite{Brink:2000ag,Zinoviev:2008ve,Zinoviev:2009gh} thus allowing  one to study
consistent passings of interacting theory from $\ads$ to Minkowski spacetime.

\section{Fradkin-Linetsky superalgebra}
\label{sec: FL superalgebra}

\subsection{Higher spin superalgebra $cu(2^{\cN-1},2^{\cN-1}|8)$}
\label{sec:oscillator}

Let Grassmann even variables $a_\alpha$, $b^\beta$ with $\alpha,\beta = 1,...,4$ and
Grassmann odd variables $\psi_i$ and $\bpsi^j$
with $i,j  =1,..., \cN$  satisfy the following non-vanishing
(anti-)commutation relations
\be
\label{spa}
[a_\alpha, b^\beta]_\star = \delta_\alpha^\beta\;,
\qquad
\{\psi_i, \bpsi^j\}_\star = \delta_i^j\;,
\ee
with respect to Weyl star-product
\be
\label{sp}
(F \star G)(a,b,\psi,\bpsi)=
F(a,b,\psi,\bpsi)\:(\exp\triangle)\:G(a,b,\psi,\bpsi)\,,
\ee
where
$$
\triangle = \frac{1}{2}\left(\frac{\overleftarrow{\d}}{\d
a_\alpha}\: \frac{\overrightarrow{\d}}{\d b^\alpha}-
\frac{\overleftarrow{\d}}{\d b^\alpha}\:
\frac{\overrightarrow{\d}}{\d a_\alpha} +
\frac{\overleftarrow{\d}}{\d \psi_i}\: \frac{\overrightarrow{\d}}{\d
\bpsi^i} +\frac{\overleftarrow{\d}}{\d \bpsi^i}\:
\frac{\overrightarrow{\d}}{\d \psi_i}\right).
$$
Thus we get particular Weyl -Clifford star-product algebra with elements
$F = F(a,b,\psi, \bpsi)$ \eqref{decinAnadBpsi}. The above variables are sufficient to build basis elements of
$\cN$-extended $gl(4|\cN)$ superalgebra,
\be
\label{gl}
T_\alpha{}^\beta = \half \{a_\alpha,  b^\beta\}_\star\;,
\quad
Q^i_\alpha=a_\alpha\bpsi^i\;, \quad \bar{Q}_i^\beta= b^\beta\psi_i\;,
\quad
\dps U_i{}^j=\half \{\psi_i,\bpsi^j\}_\star\;.
\ee
Basis elements $U_i{}^j$ form $R$-symmetry algebra  $U(\cN)\subset gl(4|\cN)$.
The graded supercommutator has the standard form
$
[F\,,G\}_\star
=F \star G - (-1)^{\pi(F)\pi(G)}G \star F \,,
$
where the $Z_2$ grading $\pi$ is defined by
\be
\label{parity}
F(-a,-b,\psi,\bpsi)=(-1)^{\pi(F)}F(a,b,\psi,\bpsi)\,,\qquad
\pi(F)=\mbox{0 or 1}.
\ee
Factoring out an ideal of $gl(4|\cN)$ generated by the central
element
\be
\label{centrN}
N = a_\alpha
b^\alpha -\psi_i\bpsi^i\;
\ee
yields subalgebra $sl(4|\cN)\subset gl(4|\cN)$ and the $\ads$ superalgebra $su(2,2|\cN)$ is
defined as a real form of $sl(4|\cN)$ singled out by the reality conditions defined below.

Higher spin extension of $su(2,2|\cN)$ introduced in \cite{Fradkin:1989yd} under the name $shsc^\infty(4|\cN)$
and called $cu(2^{\cN-1},2^{\cN-1}|8)$ in \cite{Vasiliev:2001zy} is associated with the star product algebra of all polynomials
$F(a,b,\psi,\bpsi)$ satisfying the condition \cite{Fradkin:1989md,Sezgin:2001zs,Vasiliev:2001zy}
\be
\label{dpoc}
[N,F]_\star=0\;.
\ee
Thus, Fradkin-Linetsky superalgebra is spanned by star-(anti)com\-mutators of
the elements of the centralizer of $N$ in the Weyl-Clifford star product algebra.
The above commutator  can be equivalently cast into the form
\be
\label{dpoc2}
[N,F]_\star = (N_a-N_b +N_\psi-N_{\bpsi})F\;,
\ee
where $N_{a,b}$ and $N_{\psi, \bpsi}$ are Euler operators \eqref{Na},\eqref{Nb},\eqref{Npsi}.
Then condition \eqref{dpoc} is  represented as
\be
\label{dpoc3}
(N_a +N_\psi) F = (N_b+N_{\bpsi})F\;,
\ee
so it follows that an element
$F\in cu(2^{\cN-1},2^{\cN-1}|8)$ depends on equal numbers of
even and odd  variables with upper and lower indices.
Expanding out  elements $F(a,b, \psi, \bpsi)$ with respect to both even and odd
variables yields expression \eqref{decinAnadBpsi}.
From \eqref{dpoc3} it follows that total numbers of upper and
lower indices of expansion coefficients coincide.
It is worth to comment that
expression \eqref{dpoc2} is in fact an adjoint star product realization of Howe dual
$gl(1)$ basis element $Z^0$ \eqref{Z0}.

To single out an appropriate real form of the complex higher spin
algebra $cu(2^{\cN-1},2^{\cN-1}|8)$ we impose reality conditions in the following
way \cite{Vasiliev:2001zy}. Introduce an involution $\dagger$ defined by the relations
\be
\label{inv}
({a}_\alpha )^\dagger  = i b^\beta C_{\beta
\alpha}\,,\quad ({b}^\alpha )^\dagger = i C^{\alpha\beta} a_\beta
\,,
\qquad
(\psi_i)^\dagger = \bpsi^i\;,\quad
(\bpsi^i)^\dagger = \psi_i\;,
\ee
where
$C_{\alpha\beta}$ and
$C^{\alpha\beta}$ are some real antisymmetric
matrices defining complex conjugation \eqref{complex},\eqref{conjmatr}, cf. \eqref{auxinv}.
An involution reverses an order
of product factors and conjugates complex numbers
$(F\star G)^\dagger = G^\dagger\star F^\dagger$, $
(\mu F)^\dagger = \mu^* F^\dagger$, $\mu
\in {\bf C}$,
where $*$ denotes complex conjugation. The
involution $\dagger$ leaves invariant the defining star product commutation relations
(\ref{spa}) and satisfies $(\dagger
)^2 = Id$. The action (\ref{inv}) of
$\dagger$ extends to an arbitrary element $F$ of the star product
algebra.

Using the  involution
$\dagger$ enables one to define a real form of the Lie
superalgebra built by virtue of a graded commutators of elements by imposing the condition
\cite{Vasiliev:1986qx}
\be
\label{reco}
F^\dagger = -i^{\pi(F)}  F\,.
\ee
This condition defines the real higher spin algebra $cu(2^{\cN-1},2^{\cN-1}|8)$
\cite{Vasiliev:2001zy}. It contains the ${\cal N}$ extended
$AdS_5$ superalgebra $su(2,2|\cN)$ as its maximal finite-dimensional subalgebra.

\subsection{Factorized higher spin superalgebra $hu_0(2^{\cN-1},2^{\cN-1}|8)$}

Superalgebra $cu(2^{\cN-1},2^{\cN-1}|8)$ is not simple and contains
infinitely many ideals $I_{P(N)}$, where $P(N)$ is any
star-polynomial of the central element $N$, spanned by the elements of the form
$\{x\in I_{P(N)} : x=P(N)\star  F,\;\; F\in cu(2^{\cN-1},2^{\cN-1}|8)\}$
\cite{Fradkin:1989md}. There are different quotient superalgebras
\be
cu(2^{\cN-1},2^{\cN-1}|8)/I_{P(N)}\;.
\ee
In particular, one may consider
maximally factorized superalgebra with $P(N) = N$. In Ref. \cite{Vasiliev:2001zy} a real form
of this quotient algebra singled out by conditions \eqref{reco}  has
been denoted as $hu_0(2^{\cN-1},2^{\cN-1}|8)$.

We note that the element $N$ is in fact the basis element $P^+$ of $gl(2)$ algebra realized
by \eqref{fermionsl(2)} on the linear space of $cu(2^{\cN-1},2^{\cN-1}|8)$ superalgebra.
It follows that factoring out $N\equiv P^+$ leaves supertraceless elements only, \textit{i.e.},
\be
P^- F(a,b, \psi,\bpsi) = 0\;,
\ee
and therefore $hu_0(2^{\cN-1},2^{\cN-1}|8)$ superalgebra is spanned by elements
with supertraceless expansion coefficients in \eqref{decinAnadBpsi}. Put differently,
representatives of the quotient superalgebra are identified with the HW vectors
of $gl(2)$ algebra, cf. \eqref{HW}.

The quotient algebra can also be defined using  the projecting technique elaborated
in \cite{Vasiliev:2001wa,Alkalaev:2002rq,Vasiliev:2004cm}. To this end one introduces
some element $\Pi$ that satisfies the following conditions
\be
\label{DN}
\Pi \star N = N\star\Pi = 0\;,
\qquad
\Pi\star F = F\star \Pi\;,
\qquad
\forall \,F\in cu(2^{\cN-1},2^{\cN-1}|8)\;.
\ee
In particular, it implies that $\Pi$ is some function of $N$
\be
\label{M1}
\Pi = M(N)\;.
\ee
Obviously, the second condition in  \eqref{DN} is satisfied and one can explicitly check that
the first condition \eqref{DN} reduces to the following differential equation
\be
x M^{\prime\prime}(x)  - (\cN-4)M^{\prime} - 4 x M = 0\;,
\ee
where $x$ is an indeterminate variable, and $M^{\prime}(x)$, $M^{\prime\prime}(x)$
are the first and the second derivatives of $M(x)$. For $\cN\neq 4$ we obtain
that the above equation is solved by
\be
\label{formalser}
M(x) = \sum_{n=0}^\infty\, \frac{2^n}{(2n)!!\, (2n+3-\cN)!!}\, x^{2n}\;,
\ee
while for the exceptional case $\cN=4$ we find that
\be
M(x) = e^{2x}\;.
\ee
The simple form  of $\Pi$
in the case of $\cN=4$ may be traced back to that $su(2,2|\cN)$ is not simple and possesses
an additional ideal to be factored out to obtain $psu(2,2|4)$. It follows that
its higher spin extension $hu_0(8,8|8)$ is not simple as well.
We hope to consider this issue in more detail  elsewhere.

\subsection{Gauging  $cu(2^{\cN-1},2^{\cN-1}|8)$ superalgebra}
\label{sec:gauging}

The gauging procedure introduces $cu(2^{\cN-1},2^{\cN-1}|8)$ as local symmetry in the corresponding
higher spin model. According to  a general analysis of \cite{Vasiliev:1986qx} we consider basis
elements $e_I$ of Lie superalgebra $cu(2^{\cN-1},2^{\cN-1}|8)$  with definite parities $\pi(e_I) = 0,1$.
Then one defines gauge connections of $cu(2^{\cN-1},2^{\cN-1}|8)$ as 1-forms
$\Omega = dx^\mu \Omega_\mu^I \,e_I$. Their parities
coincide with those of the basis elements,  $\pi(\Omega_\mu^I) = \pi(e_I) = 0,1$.
Gauge transformation and curvature are defined in a standard fashion
\be
\label{curvatura}
R_{\mu\nu} = \d_\mu\Omega_\nu - \d_\nu\Omega_\mu  + [\Omega_\mu, \Omega_\nu]_\star\;,
\ee
and
\be
\label{gotr}
\delta \Omega_\mu = D_\mu \xi \equiv \d_\mu \xi + [\Omega_\mu, \xi]_\star\;,
\quad
\delta R_{\mu\nu} = [R_{\mu\nu}, \xi]_\star\;.
\ee
Here brackets $[\cdot, \cdot]_\star$ denote commutator and it is assumed that basis elements $e_I$
commute with gauge connections. On the other hand, gauge connections commute as
\be
\Omega_\mu^I \,\Omega_\nu^J  = (-)^{\pi(e_I)\pi(e_J)} \Omega_\nu^J\, \Omega_\mu^I\;,
\ee
in accordance with boson-fermion spin statistics. Thus we obtain that  gauge fields associated with $cu(2^{\cN-1},2^{\cN-1}|8)$ are 1-forms \eqref{sreda}
satisfying
\be
\Omega_{\mu\,\beta_1 ... \beta_n}^{\;\;\alpha_1 ... \alpha_m}\;
\Omega_{\nu\,\beta_1 ... \beta_k}^{\;\;\alpha_1 ... \alpha_l}
=(-)^{(m+n)(k+l)}\;
\Omega_{\nu\,\beta_1 ... \beta_k}^{\;\;\alpha_1 ... \alpha_l}\;
\Omega_{\mu\,\beta_1 ... \beta_n}^{\;\;\alpha_1 ... \alpha_m}\;.
\ee
$R$-symmetry algebra indices are implicit here. Let us note that constructing gauge superalgebra
$cu(2^{\cN-1},2^{\cN-1}|8)$   involves two mutually commuting Grassmann algebras, one formed by gauge connections and
another  formed by auxiliary variables  themselves.
It is worth noting that the above definition replaces a graded commutator  by usual
commutator. This happens because for $cu(2^{\cN-1},2^{\cN-1}|8)$ Lie superalgebra we chosen the so-called
first-class Grassmann shell   \cite{Berezin:1987wh} (see also \cite{Vasiliev:1986qx}).

\subsection{$\cN=2$ higher spin supermultiplets}
\label{sec:supermult}

From now on we set  $\cN=2$  and   confine ourselves to the case
of $cu(2,2|8)$ superalgebra.
Expanding out an arbitrary element of $cu(2,2|8)$ with respect
to Grassmann odd variables one obtains
\be
\label{Fgrassdec}
\ba{ccc}
F=F_{e_1} +F_{o_{11}}^i\, \psi_i + F_{o_{12}}{}_i\, \bpsi^i +
F_{e_{21}}\, (\epsilon^{mn}\psi_m \psi_n) + F_{e_{22}}\,
(\epsilon_{mn}\bpsi^m \bpsi^n) +F_{e_{31}}\, \psi_k\bpsi^k & &
\\
\\
\dps
+F_{e_{32}}{}_i{}^j\, \psi_j\bpsi^i
+ F_{o_{21}}^i\;  \psi_i (\psi_k \bpsi^k) + F_{o_{22}}{}_i \;
\bpsi^i (\psi_k\bpsi^k)
+  F_{e_4}(\psi_k\bpsi^k) (\psi_m\bpsi^m)\;. & &
\ea
\ee
Here expansion coefficients are $F_{e,\,o}  = F_{e,\,o}(a,b)$, subscripts $e$ (even) and  $o$ (odd)
indicate bosons and fermions,
while their indices enumerate different  fields of the supermultiplet.
Expansion coefficients  $F_{e_{32}}{}_j^i$ are traceless $F_{e_{32}}{}_i{}^i = 0$.
Fields $F_{e,\,o}(a,b)$ do not necessarily have
equal numbers of $a_\alpha$ and $b^\beta$, so $(N_a-N_b)F_{e,\,o}(a,b) = p \,F_{e,\,o}(a,b)$, where
$p=0,1,2$.

Expanding out $F(a,b)$ in  $a_\alpha$ and $b^\beta$
yields traceful  coefficients, \textit{i.e.},
$F_{\;\beta_1 ... \beta_m\gamma}^{\alpha_1 ... \alpha_n\gamma}=0$, and therefore
they decompose into a collection of traceless components.
Namely, for any fixed $n$ and $m$, a multispinor
$F_{\beta_1 ... \beta_m}^{\alpha_1 ... \alpha_n}$ decomposes into the set of
irreducible traceless components
$F^{\prime\,\alpha_1 ... \alpha_k}_{\;\;\beta_1 ... \beta_k}\,,$
with all $k+l\leq m+n$, $k-l = m-n$, $k\geq 0$, $l\geq 0$.

It follows from \eqref{Fgrassdec} that the spectrum  of $cu(2,2|8)$ gauge fields
is represented by the following sum
\be
\label{Vs}
\Omega =:
\sum_{k=0}^{\infty}\sum_{s=2}^{\infty}\;D^{(k)}_{[1]}(s) \oplus {D}^{(k)}_{[2]}(s-\frac{1}{2})\oplus
D^{(k)}_{[4]}(s-1) \oplus D^{(k)}_{[1]}(s-1,1)
\oplus
{D}^{(k)}_{[2]}(s-\frac{3}{2})\oplus D^{(k)}_{[1]}(s-2)\,,
\ee
where $D^{(k)}(s_1, s_2)$ denotes a $k$-th copy of spin-$(s_1, s_2)$ unitary irreducible
representation of $su(2,2)$ \eqref{degreeB}.
Numbers in square brackets denote dimensions of $R$-symmetry algebra $u(2)$ representations. We note that
the difference between highest and lowest spins in a supermultiplets equals 2 and highest spin field in
the supermultiplet is always bosonic.
Using formula \eqref{degreeB} one can explicitly verify a balance of bosonic and
fermionic degrees of freedom.

By way of an example let us consider $s=2$ (graviton)
supermultiplet. Modulo infinite degeneracy its field content is
given by $(2_{[1]},\frac{3}{2}_{[2]}, 1_{[4]}, (1,1)_{[1]},
\frac{1}{2}_{[2]},0)$. We stress that $(1,1)_{[1]}$
representation corresponds to massive not massless antisymmetric
field $B_{\mu\nu}$ \cite{Romans:1985ps}. Spin $s=3$ (hypergraviton) supermultiplet is given by
$(3_{[1]},\frac{5}{2}_{[2]}, 2_{[4]}, (2,1)_{[1]},
\frac{3}{2}_{[2]},1_{[1]})$. It is this supermultiplet where a "hook" field appears
for the first time. It is worth to comment that $\cN=3$ supermultiplet
contains the same spin fields as $\cN=2$ supermultiplet
but there appears also a fermionic "hook" field. Spin-$(s,2)$
field appears in $\cN=4$ supermultiplet. Generally,
it follows from \eqref{dpoc2} that a value of the second spin is given  by $s_2 \leq \cN/2$.

\section{A general view of  FV-type action}

\label{sec:reviewFV}

For the analysis of interactions we use perturbation expansion with the dynamical fields
$\Omega_1$ treated as fluctuations above
the $\ads$ background
\be
\label{go01}
\Omega = \Omega_0 +\Omega_1 \,,
\ee
where  vacuum gauge fields
$\Omega_0$ satisfy  the zero-curvature condition \eqref{zerocurv}.
Both gauge transformations and non-linear curvatures are given by formulas
\eqref{curvatura} and \eqref{gotr}.
Since $R(\Omega_0)=0$, we
have $R=R_1 +R_2\,,$ where
\be
\label{R11}
R_1 = d\Omega_1
+\Omega_0 \star\wedge \Omega_1 +  \Omega_1\star\wedge
\Omega_0\,,
\qquad
R_2 = \Omega_1 \star \wedge  \Omega_1\,.
\ee
It follows that  linearized curvatures $R_1$ are of the first order in fluctuations  while
$R_2$ contain their quadratic combinations. Gauge transformations for the first order
fields are given by
\be
\label{firstorder}
\delta \Omega_1 = D_0 \xi + [\Omega_1,\xi ]_\star\;,
\qquad
\delta R_1 = [R_1, \xi]_\star\;.
\ee
Let us note that the lowest order part of the above gauge transformation
has the form \eqref{hstr}, \eqref{curvinar}.

Higher spin gravitational interactions in the cubic approximations
can be described by FV-type action functional
\be
\label{acta}
S(\Omega)=\frac{1}{2}{\cal A}\big(R(\Omega),R(\Omega)\big)\,,
\ee
where $R(\Omega)$ are 2-form curvatures  associated to gauge fields of higher spin
superalgebra. ${\cal A}(F,G)={\cal A}(G,F)$ is a bilinear symmetric inner product of the type \eqref{actschem}
defined for any differential $2$-forms $F$ and $G$
(for more details see \cite{Vasiliev:2001wa,Alkalaev:2002rq,Alkalaev:2007bq,Sorokin:2008tf}).

It is important that the above action is to be supplemented by off-shell
constraints \eqref{constr},
\be
\label{offshellconstr}
\hat \Upsilon(R_1) = 0\;.
\ee
In other words, to maintain gauge invariance of the action in
the cubic approximation one has to add constraints which are some linear combinations
of the linearized higher spin curvatures. The constraints express all extra fields via derivatives
of physical fields as in \eqref{coco}.

Before explicitly constructing cubic order theory for $\ads$ higher spin fields it will be useful to consider the
general scheme
of how to prove establish gauge invariance of the FV-type coupling. For a more detailed discussion see
\cite{Fradkin:1987ks,Fradkin:1986qy,Vasiliev:2001wa,Alkalaev:2002rq}.
The gauge invariance of the action can be achieved by attributing to fields $\Omega_1$ a
suitable transformation law. Indeed, the action can be made invariant provided $\Omega_1$
transform as
\be
\label{deformedtransfo}
\delta\Omega_1  = D \xi + \Delta(R, \xi)\;,
\ee
where $\Delta(R, \xi)$ denotes some $R$-dependent deformation of the
original transformation law  \eqref{gotr} such that $\Delta(0,\xi) = 0$. These deformations are the so-called improved diffeomorphisms
which are  intrinsic to all theories containing propagating gravity \cite{Jackiw:1978ar}.
In what follows we denote the undeformed transformation \eqref{gotr}
as $\delta^{alg} \Omega_1$ thus emphasizing its origin in $cu(2,2|\,8)$ superalgebra.

Within the perturbation scheme both the action and the gauge transformations are expanded as
\be
\ba{l}
\dps
S(\Omega_1) = S_2(\Omega_1) + S_3(\Omega_1) + ... \;,
\\
\\
\dps
\;\;\delta \Omega_1 = \delta_0 \Omega_1 + \delta_1 \Omega_1 + ... \;.
\ea
\ee
Here zeroth order transformation $\delta_0 \Omega_1$ is given by expression \eqref{hstr}. Since quadratic action
is invariant under linearized transformations, $\delta_0 S_2 = 0$, it follows that the action in the cubic approximation
stays invariant against deformed transformations \eqref{deformedtransfo} if
\be
\label{fg}
\delta^{alg} S + \Delta S_2 +... = 0\;,
\ee
where the dots stand for higher order corrections $\cO(\Omega_1^3\xi)$.
Recalling that the quadratic action does not depend
on extra fields and auxiliary fields are expressed via derivatives of physical fields, one obtains
$\Delta S_2 =\dps \frac{\delta S_2}{\delta \omega^0}\Delta \omega^0$, where $\omega^0$ denote
physical
fields.\footnote{Recall that the physical field $\omega^0$ is the Lorentz field $\omega^t$
\eqref{omegas} at $t=0$ and for hook fields it is identified with $\text{Re}\, \omega^0$, see
the discussion in the end of Section \bref{sec:gaugefieldsaux}. } Let us note that
both $\dps \frac{\delta S_2}{\delta \omega^0}$
and deformation $\Delta$ are
proportional to
linearized curvature $R_1$. According to \eqref{fg}  a deformation of the
original gauge transformation \eqref{deformedtransfo} guaranteeing  the cubic order
gauge invariance of the action
does exists provided that $\delta^{alg} S$ is a definite bilinear combination of
curvatures and the gauge parameter $\xi$, \textit{i.e.},
\be
\delta^{alg} S \sim R_1\, R_1\, \xi+... \;.
\ee
We observe that up to higher order corrections  $\delta^{alg} S$ vanishes
provided that free field equations are fulfilled, $\dps \frac{\delta S_2}{\delta \omega^0}=0$.
Using constraints \eqref{offshellconstr} and Propositions \bref{fronprop} and \bref{hookprop} one
reformulates the
gauge invariance condition in the cubic approximation as follows
\be
\label{condinvariance}
\delta^{alg} S\, \Big|_{R_1 = C} = 0\;,
\ee
where $C$ are generalized Weyl tensors.
In particular, fulfilling the invariance condition
\eqref{condinvariance} guarantees the existence of an appropriate deformation $\Delta$
of the algebraic gauge transformation law for the physical field.

Note that algebraic gauge variations of auxiliary and extra fields are also deformed
but these corrections are irrelevant for the action variation  in the cubic approximation.
Indeed, auxiliary and extra fields contribute both to the cubic action and to
constraints \eqref{offshellconstr} but due to  the extra field decoupling condition  they enter the
 action only in trilinear combinations $\Omega_1 \Omega_1 \Omega_1$. The cubic
approximation variation of the action  is given by bilinear combinations
$\Omega_1\Omega_1$. It immediately follows that first order corrections of the gauge
transformation law for  auxiliary and extra fields are irrelevant in the gauge variation
of the action and it is sufficient to know just their zeroth order part.
On the other hand, because linearized curvatures $R_1$ transform homogeneously
\eqref{firstorder} the gauge variation of  constraints \eqref{offshellconstr} is of the
first order in $\Omega_1$. Therefore,  to maintain gauge invariance of the constraints
one deforms extra field  gauge transformations  by
terms linear in $\Omega_1$. However, they do not contribute  to the variation
of the action.

The above consideration provides  a general scheme  of how to achieve a gauge invariance in
FV-type theories. However,
higher spin  models in question possess several peculiar features as local
supersymmetry and an infinite degeneracy of the
spectrum. It follows that the action should fulfill  additional conditions.

\begin{itemize}

\item \textit{$R$-symmetry invariance.} The $\cN=2$ superalgebra $cu(2,2\,|8)$ is invariant under
global $u(2)$ rotations of
supercharges \eqref{gl}. Therefore, a corresponding  field  theory should also exhibit such a global
symmetry, referred to as $R$-symmetry.

\item \textit{Factorization condition.} Superalgebra $cu(2,2|8)$ gives rise
to an infinite set of copies for a given spin field.  The factorization
condition diagonalizes a quadratic part of
the action \eqref{acta} so that different copies of the same spin field do not mix up
in the quadratic action.

\item \textit{$C$-invariance condition.} The action possesses a cyclic property with respect to the central element $N$
of $cu(2,2|\,8)$ superalgebra,
\be
\label{Cinv}
\cA(N\star F,G) = \cA(F, G\star N)\;,
\ee
where $F,G$ are $cu(2,2|\,8)$ elements   and hence they commute with $N$ \eqref{dpoc}.
\end{itemize}

\noindent In the subsequent sections we  consider each of the above conditions.
Note that the factorization condition and the $C$-invariance condition were originally
formulated in \cite{Vasiliev:2001wa} for pure bosonic theory while
their $\cN=1$ extension
was considered in \cite{Alkalaev:2002rq}.

The full action \eqref{acta} is naturally split into a sum of bosonic and fermionic parts
\be
\label{act34}
\dps
{\cal A}(F,G)= {\cal
B}(F_{e},G_{e})+ {\cal F}(F_{o},G_{o})\;,
\ee
where subscripts $e$ (even) and $o$ (odd) indicate bosonic and fermionic components of $F$ and $G$, while
$\cB$ and $\cF$ are bosonic and fermionic actions, respectively.

The fermionic part  is a sum of actions for two $su(2)$-valued totally symmetric fermions
\be
\label{fermact1}
\dps {\cal F}(F_{o},G_{o}) = \cF_1 (F_{o_1},G_{o_1}) +\cF_2(F_{o_2},G_{o_2})\;,
\ee
where
\be
{\cal F}_1(F_{o_1},G_{o_1})=\frac{1}{2}\int \hat{H}_{o_1}\wedge
G_{o_{12}}{}_i\, \wedge F_{o_{11}}^i
+\frac{1}{2}\int \hat{H}_{o_1}\wedge F_{o_{12}}{}_i\,\wedge G_{o_{11}}^i
\ee
\be
{\cal F}_2(F_{o_2},G_{o_2})=\frac{1}{2}\int \hat{H}_{o_2}\wedge
G_{o_{22}}{}_i\,\wedge F_{o_{21}}^i
+\frac{1}{2}\int \hat{H}_{o_2}\wedge F_{o_{22}}{}_i\,\wedge G_{o_{21}}^i\,.
\ee

The bosonic part is a sum of  actions for five $u(2)$-valued bosonic fields
\be
\label{bosact}
\ba{ccc}
\dps
{\cal B}(F_{e},G_{e})
=
\cB_1(F_{e_1},G_{e_1})
+{\cal B}_{31}(F_{e_{31}},G_{e_{31}}) & &
\\
\\
& & \hspace{-6.4cm} +\,{\cal B}_{32}(F_{e_{32}},G_{e_{32}})
+{\cal B}_4(F_{e_4},G_{e_4})
+{\cal B}_2(F_{e_2},G_{e_2})\;,
\ea
\ee
where each term is defined as follows. Actions for totally symmetric
fields are
\be
\ba{c}
\dps
\qquad\cB_1(F_{e_1},G_{e_1})= \int \hat{H}_{e_1}\wedge F_{e_1}\,\wedge G_{e_1}\,,
\qquad
\;\;\;\;{\cal B}_{31}(F_{e_{31}},G_{e_{31}})= \int \hat{H}_{e_{31}}\wedge F_{e
_{31}}\,\wedge G_{e_{31}}\,,
\\
\\
\dps
{\cal B}_{32}(F_{e_{32}},G_{e_{32}})= \int \hat{H}_{e_{31}}\wedge F_{e_{31}}{}_i{}^j\,\wedge G_{e_{31}}{}_j{}^i\,,
\qquad
{\cal B}_4(F_{e_4},G_{e_4})= \int \hat{H}_{e_4}\wedge F_{e_4}\,\wedge G_{e_4}\,,
\ea
\ee
while the  action for  non-symmetric fields is
\be
\label{B2}
{\cal B}_2(F_{e_2},G_{e_2})= \half \int \hat{H}_{e_2}\wedge F_{e_{22}}\,\wedge G_{e_{21}}
+
\half \int \hat{H}_{e_2}\wedge G_{e_{22}}\,\wedge F_{e_{21}}\,.
\ee
From now on the symbol of exterior product $\wedge$ will be systematically
omitted. By construction, all the above actions are invariant under $R$-symmetry transformations $u(2)$.
They are of the type \eqref{act} defined by operators
$\hat H_{e,\,o}=\hat H_{e,\,o} (E)$  \eqref{H} which depend on dynamical gravitation field
described by the frame $E^{\alpha\beta}$. To construct the cubic order action
we use the following anzats for operators $\hat H_{e,\,o}=\hat H_{e,\,o} (E)$.
Namely, we set a part of coefficients or their linear combinations to zero
\be
\label{coefother1}
\beta_{e}(p,q) = 0\;,
\qquad
\zeta_{e}(p,q) = -\gamma_{e}(p,q)\;,
\ee
for spin-$s$ bosonic fields, and
\be
\label{coefother2}
\beta_{e,o}(p,q) = 0\;,
\qquad
\gamma_{e,o}(p,q) = 0\;,
\ee
for spin-$s$ fermionic fields and spin-$(s,1)$ bosonic fields. Note that the above choice is consistent
with the quadratic action coefficients \eqref{K1}-\eqref{K3}.

It is important to comment that
describing gauge fields as differential forms and using the compensator mechanism
that makes $su(2,2)$ symmetry manifest guarantees that the full
action \eqref{act34} is explicitly $su(2,2)$ covariant and diffeomorphism
invariant. Note that we treat gravitational fields appearing in the full action in two different setups, as the frame
field $E^{\alpha\beta}$ that explicitly enters operators $\hat H_{e,\,o}=\hat H_{e,\,o} (E)$
and gauge connection $\Omega^{\alpha\beta}$ of $su(2,2)\subset cu(2,2|8)$.
As a result, the gauge variation  $\delta^{alg}S$ of the full action \eqref{act34}
involves two types of terms resulting from varying operators
$\hat H_{e,\,o} (E)$ and curvatures $R(\Omega_1)$. The invariance of the first
type results from the explicit $su(2,2)$ covariance and diffeomorphism
invariance of the whole setup. The invariance of the second type gives rise to the
condition \eqref{condinvariance} which now takes the form
\be
\label{skoooro}
\cA(R_1, [R_1, \xi]_\star) \approx 0\;.
\ee
here $\cA$ is given by \eqref{act34} and  $\approx$ means that all linearized
curvatures $R_1$ are replaced by generalized Weyl tensors according to
Propositions \bref{fronprop} and \bref{hookprop}. Gauge parameter $\xi\in cu(2,2|8)$
is arbitrary.

The above discussion of the gauge invariance in the cubic approximation
is valid for a higher spin  model with $cu(2,2|8)$ local symmetry  but the same
methods are also  applied for a reduced system governed by
factorized superalgebra $hu_0(2,2|\,8)$. To build a reduced model we use
the approach elaborated for  $\cN=0$ pure bosonic system in
\cite{Vasiliev:2001wa} and for $\cN=1$  system in \cite{Alkalaev:2002rq}
which consists of inserting the projecting operator  $\Pi$ \eqref{DN}
into the action of  $cu(2,2|8)$ system as
\be
\label{bformc}
{\cal A}(F,G) \to {\cal A}_0 (F,G)=  {\cal A}(F,\Pi \star G)\,,
\ee
where ${\cal A}(F,G)$ is given by  \eqref{act34}. Then ${\cal A}_0 (F,G)$ defines
an action of the reduced model.
Because  the projecting operator $\Pi(N)$ is some fixed function of $N$ \eqref{formalser} it follows
that the $C$-invariance condition guarantees
\be
{\cal A}(F, \Pi \star G)={\cal A}(F\star\Pi , G)\,,
\ee
so the bilinear form in the action with $\Pi$ inserted remains symmetric.
The idea is that all terms in $F$ and $G$ proportional to $N$ do not contribute to
the action (\ref{bformc}) which therefore is defined on the
quotient subalgebra $hu_0(2,2|8)$. Note that ${\cal A}_0(F,G)$  is well-defined as a functional of
polynomial functions $F$ and $G$ because for polynomial $F$ and
$G$ only a finite number of terms in the expansion of $\Pi$ in
auxiliary variables contributes. The explicit expression for $\cN=2$ projecting
operator  $\Pi$ is given in Section \bref{sec:summary}.

\subsection{Summary of results}
\label{sec:summary}

In this section we list all the coefficients in the action for $cu(2,2|8)$ model.
\begin{itemize}
\item  Spin-$(s_1,0)$ sector is given by
\be
\ba{c}
\dps
\alpha_{e_1}(p,q) = 2 \gamma_{e_1} (p+q)
- \half\Phi_0 \int^1_0 d\tau \, {\rm Res}_\nu  \nu e^{\half
(-\nu^{-1} +\nu(\tau p +q)) }
\\
\\
\dps
\gamma_{e_1}(p)= -\zeta_{e_1}(p) = -\frac{\Phi_{0}}{4}\int^1_0 d\tau \tau \, {\rm Res}_\nu  \nu e^{\half
(-\nu^{-1} +\nu \tau p )}
\ea
\ee
\be
\ba{c}
\dps
\alpha_{e_{31}}(p,q) = \half \alpha_{e_1}(p,q)
\qquad
\gamma_{e_{31}}(p) =\half \gamma_{e_{1}}(p)
\\
\\
\dps
\alpha_{e_{4}}(p,q) = \frac{1}{4} \alpha_{e_1}(p,q)
\qquad
\gamma_{e_{4}}(p) =\frac{1}{4}\gamma_{e_{1}}(p)

\ea
\ee
\be
\ba{c}
\dps
\alpha_{e_{32}}(p,q) = 2 \gamma_{e_{32}} (p+q)
- \frac{1}{8}\Phi_{0} \int^1_0 d\tau \, {\rm Res}_\nu \nu e^{\half
(-\nu^{-1} +\nu(\tau p +q)) }
\\
\\
\dps
\gamma_{e_{32}}(p) = - \zeta_{e_{32}}(p)= -\frac{\Phi_{0}}{16}\int^1_0 d\tau \tau \, {\rm Res}_\nu
\nu e^{\half (-\nu^{-1} +\nu \tau p )}
\ea
\ee
According to \eqref{coefother1} all coefficients $\beta_{e_1}(p,q) =\beta_{e_{31}}(p,q)=\beta_{e_{32}}(p,q)=\beta_{e_4}(p,q) =0$.

\item Spin-$(s_1,1)$ sector is given by
\be
\ba{c}
\dps
\alpha_{e_2}(p,q) =\zeta_{e_2}(p,q)+\frac{\Phi_{0}}{q}\, \int_0^1 d \tau {\rm Res}_\nu \nu^{-1}e^{\half(-\nu^{-1}+\nu (\tau p+q))}
\\
\\
\dps
\zeta_{e_2}(p,q) = -\frac{\Phi_{0}}{q(p+q)}\,
\int_0^1 d \tau {\rm Res}_\nu \nu^{-1}e^{\half(-\nu^{-1}+\nu \tau (p+q))}
\ea
\ee
According to \eqref{coefother2} coefficients $\beta_{e_2}(p,q) = 0$ and $\gamma_{e_{2}}(p,q) = 0$.

\item Spin-$(s_1,\half)$ sector is given by
\be
\ba{c}
\dps
\alpha_{o_1}(p,q)=\zeta_{o_1}(p,q)+
\frac{\Phi_{0}}{2q}\int_0^1 d\tau \,{\rm Res}_\nu\,
e^{\frac{1}{2}(-\nu^{-1}+\nu (p \tau+q))}
\\
\\
\dps
\zeta_{o_1}(p,q)=-\frac{\Phi_{0}}{2q(p+q)}\int_0^1 d\tau \,{\rm Res}_\nu\,
e^{\frac{1}{2}(-\nu^{-1}+\nu (p+q) \tau)}
\ea
\ee
\be
\alpha_{o_2}(p,q) = \frac{1}{4}\alpha_{o_1}(p,q)
\qquad
\zeta_{o_2}(p,q) = \frac{1}{4}\zeta_{o_1}(p,q)
\ee

According to \eqref{coefother2}  coefficients $\beta_{o_1}(p,q)=\beta_{o_2}(p,q) = 0$ and
$\gamma_{o_1}(p,q)=\gamma_{o_2}(p,q)=0$.

\end{itemize}

\noindent Here  $\Phi_0$ is an arbitrary  factor
properly normalized in terms of the
cosmological constant $\lambda$ and the gravitational constant $\kappa$.

The action of the reduced $hu(2,2|8)$ model  is defined according to
\eqref{bformc}, where the form of the projecting operator is read off from the general
expression \eqref{formalser} at $\cN=2$
\be
\Pi(N)= \sum_{n=0}^\infty\, \frac{2^n}{(2n)!!\, (2n+1)!!}\, N^{2n}\;.
\ee


\section{Calculation of gauge invariance}
\label{sec:calculations}

The novel feature of $\cN=2$ analysis compared to $\cN=0,1$ case is the
appearance of "hook" fields. In this section we study the invariance condition \eqref{skoooro}
giving particular emphasis to calculations involving   fields of "hook" symmetry type.
Our analysis of the cubic order interaction vertices  is heavily based on the technique elaborated in
the previous
papers on $\cN = 0,1$ FV-type theory \cite{Vasiliev:2001wa,Alkalaev:2002rq}.
In particular, we do not repeat here  calculations related to
totally symmetric fields and use results obtained in \cite{Vasiliev:2001wa,Alkalaev:2002rq}.

\subsection{Factorization condition for "hook" fields}
\label{sec:factorization}

We begin by noting that due to (super)traces of $cu(2,2|\,8)$ gauge fields
$su(2,2|2)$ supermultiplets
are not irreducible and decompose into (super)traceless components (see Section \bref{sec:supermult}).
Having in mind \eqref{HW} we call a gauge field  $\Omega(a,b,\psi,\bpsi|x)$ supertraceless if it fulfills the following algebraic constraint
\be
\label{STR}
P^-\, \Omega(a,b,\psi,\bpsi|x) =0\;,
\ee
where operator $P^-$ is given by \eqref{fermionsl(2)}.
It follows that using operators $P^-$ and $P^+$ allows one to decompose
any element $\Omega(a,b,\psi,\bpsi|x)$ of $cu(2,2|8)$ superalgebra into irreducible
$su(2,2|2)$ supermultiplets as
\be
\label{STRD}
\Omega(a,b,\psi,\bpsi|x) =
\sum_{k=0}^{\infty}\sum_{s_1=2}^\infty\,\chi(k,\,s_1)\,(P^+)^k
\;\Omega^{k,\,s_1}(a,b,\psi,\bpsi|x)\,,
\ee
where $\chi(k,\,s_1)$ are  arbitrary coefficients, $s_1$ denotes the
highest integer spin in a
supermultiplet and $\Omega^{k,s_1}$ are
supertraceless \eqref{STR}.
The supertraceless decomposition  can be equivalently rewritten (modulo finite field redefinitions)
in the manifest $su(2,2)$ fashion with all multispinors being traceless rather than
supertraceless
\be
\label{tr1}
\Omega(a,b|x) = \sum_{k=0}^{\infty}\sum_{s_1=2}^\infty\,
v(n,s_1)\;(T^+)^n \;\Omega_{}^{n,s_1}(a,b|x)\;,
\ee
where $v(n,s_1)$ are  arbitrary  coefficients
and $\Omega_{}^{n,s_1}(a,b|x)$ describe an $n$-th copy of irreducible
field of a given spin $(s_1,s_2)$ \eqref{IrreducibleSpin}. Note that $s_2 = 0,\frac{1}{2},1$
is implicit in the above decompositions. The decomposition analogous to \eqref{tr1}
is valid for the curvatures
\be
\label{tr32}
R(a,b|x) = \sum_{n,\,s_1=0}^{\infty} \,
v(n,s_1)\;(T^+)^n \;R_{}^{n,s_1}(a,b|x)\;,
\ee
where $R_{}^{n,s_1}(a,b|x)$ are associated with
irreducible fields $\Omega_{}^{n,s_1}(a,b|x)$.

The factorization condition requires
\be
\label{fc}
S_2(\Omega) = \sum_{n=0}^\infty\sum_{s_1=2}^{\infty} \,\sum_{s_2}\, S_2^{n,s_1,s_2}
(\Omega_{}^{n,s_1+2,s_2})\;,
\ee
where $S_2$ is a quadratic part of \eqref{acta} and $S_2^{n,s_1,s_2}$ is
a quadratic action
for a $n$-th copy of a given spin field (recall that it may take values in $u(2)$ irreps). The condition diagonalizes $S_2$,
\textit{i.e.} the terms containing products of the fields  $\Omega^{n,\,s_1,s_2}$ and
$\Omega^{m,\,s_1,s_2}$ with $n\neq m$ in the trace decomposition
(\ref{tr1}) should all vanish. Note that normalization coefficients $v_{n}(T^0)$
in expansion \eqref{tr1} can be chosen in such a way that all copies of the same spin in the  quadratic
actions enter \eqref{fc} with the same overall factor.
The factorization condition for totally symmetric fields has been
explicitly calculated in Refs. \cite{Vasiliev:2001wa,Alkalaev:2002rq}. In this section
we perform  the analogous analysis for "hook" fields.

From the above discussion it follows that the factorization condition
in the spin-$(s_1,1)$ sector is valid provided that
\be
\label{fact0}
\cB_2(F_{e_2},T^+ G_{e_2})
=
\tilde{\cB}_2(T^- F_{e_2},G_{e_2})\;,
\ee
where action $\tilde{\cB_2}$ is defined for some set of new parameters
$(\tilde \alpha_{e_2}, \tilde \zeta_{e_2})$ expressed in terms of old
parameters $(\alpha_{e_2}, \zeta_{e_2})$, see \eqref{H} and \eqref{K3}.
Then one finds that two actions differ from each other  by the following term
\be
\int Q_{e_2}(p,q) E_\alpha{}^\beta\frac{\d^2}{\d a_{2\alpha}\d b_2^\beta}({ c}_{12})^2 F_{e_{21}}(a_1, b_1) G_{e_{22}}(a_2, b_2)\;,
\ee
which is required to vanish,
\be
\label{Q=0}
Q_{e_2}(p,q)\equiv  (1+p\frac{\d}{\d p})\,\alpha_{e_2}(p,q)+(1+q\frac{\d}{\d q})\,\zeta_{e_2}(p,q) =0\;.
\ee
The new coefficients are expressed through the old ones as follows
\be
\label{tildecoef}
\ba{l}
\dps
\tilde \alpha_{e_2}(p,q) = 4\left((2+p\frac{\d}{\d p})\frac{\d}{\d p}+
(3+q\frac{\d}{\d q})\frac{\d}{\d q}\right)
\alpha_{e_2}(p,q)\,,
\\
\\
\dps
\tilde \zeta_{e_2}(p,q) = 4\left(\frac{2}{q}+(1+p\frac{\d}{\d p})\frac{\d}{\d p}+
(4+q\frac{\d}{\d q})\frac{\d}{\d q}\right)
\zeta_{e_2}(p,q)\,.
\ea
\ee
They will be further constrained by the $C$-invariance condition discussed below.
One can show that the factorization condition \eqref{Q=0} and the extra field
decoupling condition \eqref{exdc} are compatible and the solution is
given by \eqref{K3}. Quite analogously one considers totally symmetric fields
and proves that the coefficients are fixed by the factorization and extra field
decoupling conditions as in \eqref{K1} and \eqref{K2}, see \cite{Vasiliev:2001wa,Alkalaev:2002rq}.

\subsection{The $C$-invariance condition}

Let us discuss  the $C$-invariance condition \eqref{Cinv}.
The exact formula for $N\star F$ reads
\be
\ba{l}
\dps
N\star F =(T^+-T^-)F - F_{e_1}\,(\psi_k\bpsi^k)
- F_{o_{11}}^i\,\psi_i\,(\psi_k\bpsi^k)
- F_{o_{12}}{}_i\,\bpsi^i\, (\psi_k\bpsi^k)
\\
\\
\dps
- F_{e_{31}}\,(\psi_k\bpsi^k)\,(\psi_m\bpsi^m) - \frac{1}{2}F_{e_{31}}
- \frac{1}{4}F_{o_{21}}^i\,\psi_i
- \frac{1}{4}F_{o_{22}}{}_i\,\bpsi^i
- \frac{1}{2}F_{e_4}\,(\psi_m\bpsi^m)\;.
\ea
\ee
where $F$ is given by \eqref{Fgrassdec}. Imposing the $C$-invariance condition results in the
mutual conjugation
of the  trace creation operator $T^+$ and trace annihilation operator $T^-$
with respect to the inner product $\cA$:
\be
\label{conjugat}
\cA(T^{\pm}F, G) = - \cA(F, T^\mp G)\;,
\ee
while the relative coefficients between different type actions are fixed as
\be
\label{Cinvnumbers}
\cB_1 = 2\cB_{31}  \;,
\qquad
\cB_1 = 4\cB_{4}  \;,
\qquad
\cF_1 = 4\cF_{2} \;.
\ee

In particular, condition \eqref{conjugat} implemented in the "hook" field
sector
along with the factorization condition yields
additional relations for coefficients \eqref{tildecoef},
\be
\label{newold}
\alpha(p,q)+ \tilde\alpha(p,q) = 0\;,
\qquad
\zeta(p,q)+ \tilde\zeta(p,q) = 0\;.
\ee
It is worth noting that the factorization condition is implemented on the free field
level only while the C-invariance conditions is valid for the non-linear action as well. In particular,
we see that  condition \eqref{conjugat} for free fields is a stronger version of the
factorization condition. Also,  conditions \eqref{Cinvnumbers} for free fields
are too restrictive because  they relate  normalization constants
in front of different spin quadratic actions.

The $C$-invariance condition  also implies that it is
sufficient to consider the invariance condition \eqref{condinvariance} only for the fields
satisfying  the tracelessness condition \eqref{HWgauge}. Because curvatures decompose
into traceless components as \eqref{tr32} we single out the zeroth order terms in $T^+$
and denote them as
\be
\cR(a,b|x) \equiv \sum_{s_1=2}^\infty \cR^{s_1}(a,b|x)
\ee
By definition, each term in this expansion is traceless, $T^- \cR^{s_1} = 0$.
Recall that both the second spin value  $s_2 = 0,\frac{1}{2},1$ and $u(2)$ indices are
implicit here. One may explicitly prove that the invariance
condition \eqref{condinvariance} is now takes the form
\be
\label{matt}
\cA(\cR, [\cR, \xi]_\star) \approx 0\;,
\ee
where $\approx$ means that all linearized curvatures are replaced by generalized
Weyl tensors according to
Propositions \bref{fronprop} and \bref{hookprop}. The idea of the proof is to consider
the variation $\cA(R,[R,\xi]_\star)$ with curvatures decomposed according to
the trace decomposition \eqref{tr32}. Then using formula \cite{Vasiliev:2001wa}
\be
\label{TT}
T^+ F(a,b) = T^+\star F(a,b) + \big(T^- - \half\, G^0\big)F(a,b)\;,
\ee
where $T^\pm$ and $G^0$ are given by \eqref{T}, \eqref{G0}, along with the
$C$-invariance condition in the form \eqref{conjugat} enables  one reduce step by step
a degree in $T^+$ thereby ending up with pure traceless curvatures $\cR$ and new
gauge parameter $\xi \rightarrow T^+\star \xi$. More detailed exposition can be
found in \cite{Vasiliev:2001wa,Alkalaev:2002rq}.

\subsection{Cubic order gauge invariance}
\label{sec:}

Gauge transformations of $cu(2,2|\,8)$ superalgebra are defined
by $0$-form parameter $\xi= \xi(a,b,\psi,\bpsi|x)$ expanded out analogously to \eqref{Fgrassdec},
\be
\label{gaugeFgrassdec}
\ba{c}
\xi =\xi_{e_1}  +\xi_{o_{11}}^i\, \psi_i + \xi_{o_{12}}{}_i\, \bpsi^i
+ \xi_{e_{21}}\, (\epsilon^{mn}\psi_m \psi_n)
+ \xi_{e_{22}}\, (\epsilon_{mn}\bpsi^m \bpsi^n)
+\xi_{e_{31}}\, \psi_k\bpsi^k
\\
\\
\dps
+\xi_{e_{32}}{}_j^i\, \psi_i\bpsi^j
+ \xi_{o_{21}}^i\;  \psi_i (\psi_k \bpsi^k) + \xi_{o_{22}}{}_i \;
\bpsi^i (\psi_k\bpsi^k)
+  \xi_{e_4}(\psi_k\bpsi^k) (\psi_m\bpsi^m)\;.
\ea
\ee
Because  the curvatures $R(a,b, \psi, \bpsi|x)$ are transformed homogeneously  \eqref{gotr} it follows that the component
form of $\delta R(a,b, \psi, \bpsi|x)$ comprises over a hundred  terms.
In what follows we consider invariance with respect to each type of gauge transformations associated
with supermultiplet parameters \eqref{gaugeFgrassdec}, but explicit calculations are too
lengthy to present them here. Instead, we explicitly analyze the invariance with
respect to bosonic symmetry defined by $\xi_{e_1}$, while the rest of gauge invariance analysis is given
schematically just emphasizing  key points. Explicit expressions for gauge transformations
are relegated to Appendix \bref{sec:A}.

\subsubsection{Cubic order  invariance for "hook" fields}
\label{sec: main_bosonic_sym}

In this section we study the gauge invariance with respect to bosonic parameter
$\xi_{e_1} = \xi_{e_1}(a,b)$ in the "hook" field sector.
Let us note that the respective symmetry does not mix different type fields, see \eqref{purebosonicGTR}.
The gauge invariance for totally symmetric fields was analyzed in \cite{Vasiliev:2001wa,Alkalaev:2002rq}.

A general variation of the action for "hook" fields \eqref{B2} is given by
\be
{\delta \cal B}_2=  \int \hat{H}_{e_2}\; \delta R_{e_{22}}\, R_{e_{21}}
+
\int \hat{H}_{e_2}\; R_{e_{22}}\,\delta R_{e_{21}}\,.
\ee
Substituting $\delta R_{e_{21}} = [R_{e_{21}}, \xi_{e_1}]_\star$ and
$\delta R_{e_{22}} = [R_{e_{22}}, \xi_{e_1}]_\star$ from \eqref{purebosonicGTR}
we obtain

\be
\ba{c}
\dps
{\delta \cal B}_2= \int \hat{H}_{e_2}\; (R_{e_{22}}\star\xi_{e_1})\, R_{e_{21}}
-
 \int \hat{H}_{e_2}\; (\xi_{e_1}\star R_{e_{22}})\, R_{e_{21}}
\\
\\
\dps
 +
\int \hat{H}_{e_2}\; R_{e_{22}}\, (R_{e_{21}}\star\xi_{e_1})
-
\int \hat{H}_{e_2}\; R_{e_{22}}\, (\xi_{e_1}\star R_{e_{21}})\,.
\\

\ea
\ee
In order to calculate the above variation in the form \eqref{matt} we set all
traces in $R_{e_{21}}$ and $R_{e_{22}}$ to zero and for respective
traceless components use the following representation
in terms of Weyl tensors, cf. \eqref{onmassH21} and \eqref{onmassH22},
\be
\ba{l}
\dps
\cR_{e_{21}}(a,b) = {\rm Res}_{\nu}\nu^{-2} e^{\nu^{-1}a_\alpha\frac{\d}{\d c_\alpha}+\nu b^\beta\frac{\d}{\d c^\beta}}
H_2^{\gamma\rho}\frac{\d^2}{\d c^\gamma\d c^\rho} C_{e_{21}}(c)\Big|_{c=0}\;,
\\
\\
\dps
\cR_{e_{22}}(a,b) = {\rm Res}_{\nu}\nu^{-2} e^{\nu a_\alpha\frac{\d}{\d c_\alpha}+\nu^{-1} b^\beta\frac{\d}{\d c^\beta}}
H_2^{\gamma\rho}\frac{\d^2}{\d c^\gamma\d c^\rho} C_{e_{22}}(c)\Big|_{c=0}\;.
\\
\ea
\ee
We find that up to non-zero multiplicative constant  variation $\delta \cB_2$ is given by
\be
\label{puuk}
\ba{c}
\dps
\int H_5\, \bar {k}^2 {\rm Res}_\nu \, e^{\half(\nu \bar{v}_1 - \nu^{-1}\bar{u}_1)}\nu^{-2} (\nu \bar{k}+\bar{u}_2)^2 \Phi(Z)\;
C_{e_{22}}(c_1)C_{e_{21}}(c_2)\xi(a_3,b_3)
\\
\\
\dps
-\int H_5\, \bar {k}^2 {\rm Res}_\nu \, e^{\half(-\nu \bar{v}_1 + \nu^{-1}\bar{u}_1)}\nu^{-2} (\nu \bar{k}+\bar{u}_2)^2 \Phi(Z)\;
C_{e_{22}}(c_1)C_{e_{21}}(c_2)\xi(a_3,b_3)+
\ea
\ee
\be
\nonumber
\ba{l}
\\
\\
\dps
+
\int H_5\, \bar {k}^2 {\rm Res}_\nu \, e^{\half(\nu^{-1} \bar{v}_2 - \nu \bar{u}_2)}\nu^{-2} (\nu \bar{k}+\bar{v}_1)^2
\Phi(Y)\; C_{e_{22}}(c_1)C_{e_{21}}(c_2)\xi(a_3,b_3)
\\
\\
\dps
-\int H_5\, \bar {k}^2 Res_\nu \, e^{\half(-\nu^{-1} \bar{v}_2 + \nu \bar{u}_2)}\nu^{-2} (\nu \bar{k}+\bar{v}_1)^2
\Phi(Y)\; C_{e_{22}}(c_1)C_{e_{21}}(c_2)\xi(a_3,b_3)\;,
\ea
\ee
where we used the following notation
\be
\bar{k}=\frac{\d^2}{\d c_{1\,\alpha} \d c_2^\alpha}\,,
\qquad
\bar{u}_i=\frac{\d^2}{\d c_i^\alpha \d a_{3\,\alpha}}\,,
\qquad \bar{v}_i=\frac{\d^2}{\d c_{i\,\alpha} \d b_3^\alpha}\,,
\ee
and
\be
\ba{c}
\dps
Z \equiv AB =   (\nu \bar{k}+\bar{u}_2)(\nu^{-1} \bar{k}-\bar{v}_2)\;,
\qquad
Y \equiv FD =  (\nu\bar{k} + \bar{v}_1)(\nu^{-1}\bar{k} -\bar{u}_1)\;,
\ea
\ee
while the function $\Phi(Z)$ is given by
\be
\label{ExampleOfPhi}
\Phi(Z) = Z(\alpha_{e_2}(Z,-Z) - \zeta_{e_2}(Z,-Z))\;.
\ee
Quantity $H_5$ is a $5$-form defined as $H_5 = h_{\alpha}{}^\beta h_{\beta}{}^\gamma h_{\gamma}{}^\rho h_{\rho}{}^\delta h_{\delta}{}^\alpha$ \cite{Vasiliev:2001wa}.
The invariance condition \eqref{matt} requires the above variation to vanish.
Because it is legitimate to omit generalized Weyl tensors and $H_5 \bar k^2$ in the left-hand-side
of \eqref{puuk} we obtain the following equation
\be
\ba{c}
{\rm Res}_\nu \, e^{\half(\nu \bar{v}_1 - \nu^{-1}\bar{u}_1)}\nu^{-2} A^2 \Phi(AB)\;
-{\rm Res}_\nu \, e^{\half(-\nu \bar{v}_1 + \nu^{-1}\bar{u}_1)}\nu^{-2} A^2 \Phi(AB)\;
\\
\\
\dps
\;\;\;\;\;\;\;+
{\rm Res}_\nu \, e^{\half(\nu^{-1} \bar{v}_2 - \nu \bar{u}_2)}\nu^{-2} F^2
\Phi(FD)\;
-{\rm Res}_\nu \, e^{\half(-\nu^{-1} \bar{v}_2 + \nu \bar{u}_2)}\nu^{-2} F^2
\Phi(FD)\; = 0\;.
\\

\ea
\ee
Let us define a function $\tilde \Phi(A,B) = A^2 \Phi(AB)$ and rewrite the above equation as follows
\be
\label{526}
\ba{c}
{\rm Res}_\nu \nu^{-2}\Big(\, e^{\half(\nu \bar{v}_1 - \nu^{-1}\bar{u}_1)} \tilde \Phi(A,B)\;
-e^{\half(-\nu \bar{v}_1 + \nu^{-1}\bar{u}_1)}  \tilde \Phi(A,B)\;
+
\\
\\
\dps
e^{\half(\nu^{-1} \bar{v}_2 - \nu \bar{u}_2)} \tilde
\Phi(F,D)\;
-e^{\half(-\nu^{-1} \bar{v}_2 + \nu \bar{u}_2)} \tilde
\Phi(F,D)\;\Big) = 0\;.
\\

\ea
\ee
An educated guess is that the function $\tilde \Phi (A,B) = \Phi^{e_2}_0 Res_\mu (\mu^{-2}e^{\half(\mu A + \mu^{-1}B)})$, where
$\Phi^{e_2}_0$ is an arbitrary constant, is a solution to
the above equation. Indeed, substituting this function back into \eqref{526} gives
\be
\ba{c}
{\rm Res}_\nu \nu^{-2}\mu^{-2}\Big(\, e^{\half(\nu \bar{v}_1 - \nu^{-1}\bar{u}_1) + \half(\mu A + \mu^{-1}B)}\;
-e^{\half(-\nu \bar{v}_1 + \nu^{-1}\bar{u}_1)+ \half(\mu A + \mu^{-1}B)}\;
+
\\
\\
\dps
e^{\half(\nu^{-1} \bar{v}_2 - \nu \bar{u}_2)+ \half(\mu F + \mu^{-1}D)}
-e^{\half(-\nu^{-1} \bar{v}_2 + \nu \bar{u}_2)+ \half(\mu F + \mu^{-1}D)} \;\Big) = 0\;,
\\

\ea
\ee
or
\be
\ba{c}
{\rm Res}_\nu \nu^{-2}\mu^{-2}\Big(\, e^{\half(\nu \bar{v}_1 - \nu^{-1}\bar{u}_1) + \half(\mu (\nu \bar{k}+\bar{u}_2) + \mu^{-1}(\nu^{-1} \bar{k}-\bar{v}_2))}\;
-e^{\half(-\nu \bar{v}_1 + \nu^{-1}\bar{u}_1)+ \half(\mu (\nu \bar{k}+\bar{u}_2) + \mu^{-1}(\nu^{-1} \bar{k}-\bar{v}_2))}\;
+
\\
\\
\dps
e^{\half(\nu^{-1} \bar{v}_2 - \nu \bar{u}_2)+ \half(\mu (\nu\bar{k} + \bar{v}_1) + \mu^{-1}(\nu^{-1}\bar{k} -\bar{u}_1))}
-e^{\half(-\nu^{-1} \bar{v}_2 + \nu \bar{u}_2)+ \half(\mu (\nu\bar{k} + \bar{v}_1) + \mu^{-1}(\nu^{-1}\bar{k} -\bar{u}_1))} \;\Big) = 0\;.
\\

\ea
\ee
The first and the forth terms are equal to  each other under
$\nu \leftrightarrow\mu$, while the second and the third terms are
equal to each other under $\nu \leftrightarrow - \mu $.
Therefore, we conclude that the function
\be
\label{hookPhi}
\Phi(A) = \Phi_0^{e_2} A^{-2} {\rm Res}_{\mu} \Big(\mu^{-2} \exp{\half(\mu A +\mu^{-1})}\Big)\;,
\ee
where $A$ is some indeterminate variable,
solves the invariance condition in the sector of "hook" fields. As a result, we arrive at the following equation
on the coefficient functions
\be
\label{mastereq}
A(\alpha(A, -A) -\zeta(A,-A))= \Phi_0 A^{-2} {\rm Res}_{\mu} \Big(\mu^{-2} \exp{\half(\mu A +\mu^{-1})}\Big)\;.
\ee
The left-hand-side of the above equation does not vanish
at $A=0$ because the coefficient $\zeta(A,-A)$ is not necessarily polynomial and contains
poles in $A$. Contrary, the right-hand-side is polynomial but the zeroth order in $A$ is not
generally zero so the equation is consistent at $A=0$.
Let us note that though the above equation involves the coefficients which are functions
of two variables $p$ and $q$ it defines dependence on just one variable. Actually
this is due to the fact that equation \eqref{mastereq} involves a function of a single variable
$\rho(p+q)$ which defines normalization  constants  in front of quadratic actions \eqref{K1}-\eqref{K3}.

Equation \eqref{mastereq} can be cast into the following convenient integral form
\be
\alpha(A, -A) -\zeta(A,-A) = \frac{\Phi_0^{e_2}}{2} A^{-2}\int_0^1 d \tau {\rm Res}_\nu \nu^{-1}e^{\half(\nu^{-1}+\nu \tau A)}\;.
\ee
We write down the answer in terms of function
\be
\rho(p) = -\frac{\Phi_0}{2p}\, \int_0^1 d \tau {\rm Res}_\nu \nu^{-1}e^{\half(-\nu^{-1}+\nu \tau p)}\;.
\ee
It follows that the coefficient functions take the the form
\be
\zeta(p,q) = \frac{\rho(p+q)}{q}\;,
\ee
\be
\alpha(p,q) =\frac{\rho(p+q)}{q}+\frac{\Phi_0}{2q}\, \int_0^1 d \tau {\rm Res}_\nu \nu^{-1}e^{\half(-\nu^{-1}+\nu (\tau p+q))}\;.
\ee
One can explicitly check  that the above formal series satisfy the following identities
\be
\Big(p\frac{\d^2}{\d p^2}+3\frac{\d}{\d p}+\frac{1}{4}\Big)\rho(p) = 0\;,
\ee
\be
\Big((2+p\frac{\d}{\d p})\frac{\d}{\d p}+
(3+q\frac{\d}{\d q})\frac{\d}{\d q}+\frac{1}{4}\Big)\alpha(p,q) = 0\;,
\ee
which are in fact conditions \eqref{tildecoef}, \eqref{newold}. Thus it is shown that the coefficient
functions  for "hook" fields satisfy the factorization condition,
the $C$-invariance condition, extra field decoupling
condition and the invariance condition (\ref{matt}). One concludes that the action for "hook" fields
is consistently defined both on the free field and interaction levels.

\subsubsection{The remaining invariance}
\label{sec:remaining_invariance}

Gauge invariance of actions for totally symmetric bosonic and fermionic fields  with respect to $\xi_{e_1}(a,b) $
has been considered in \cite{Vasiliev:2001wa, Alkalaev:2002rq}.
The common feature
of the variation in different field sectors of the full action
is that coefficient functions $\alpha(p,q)$, $\beta(p,q)$, $\gamma(p,q)$, and $\zeta(p,q)$
in \eqref{H} appear  only through particular combinations identified with functions $\Phi(X)$
of the type \eqref{ExampleOfPhi}; exact expressions  are collected
in \eqref{fron}-\eqref{hook}.
It follows  that considering the gauge  variation is more
convenient in terms of functions $\Phi(X)$.
Taking into account the results obtained in the previous section
we list  functions $\Phi(X)$ for spin-$s_1$ fields and for spin-$(s_1,1)$ in the following manner

\be
\label{solution}
\Phi(X) = \Phi_0 \Psi(X)\;,
\qquad
\Psi(X) = X^{-2s_2} {\rm Res}_{\nu} \big(\nu^{-2s_2} \exp{\half(\nu^{-1}+\nu X)}\big)\;,
\ee
where $X$ is an indeterminate variable, normalization constants $\Phi_0$ are
arbitrary, and $s_2 = 0, \half, 1$.
This result tells us that gauge invariance with parameter $\xi_{e_1}$ fixes all coefficients inside
actions for each type of supermultiplet fields and leaves arbitrary overall constants. The remaining
gauge invariance imposes on them some linear relations so that all these constants
are expressed via a single normalization constant.

Prior discussing the remaining gauge invariance let us make the
following observation. By virtue of the $C$-invariance condition
the invariance with respect to $\xi_{e_1}(a,b)$ yields the invariance with respect to
bosonic parameters $\xi_{e_{31}}(a,b)$ and $\xi_{e_{4}}(a,b)$.
Indeed, suppose we proved invariance of the action with respect to
$\xi_{e_1}$, \textit{i.e.} the condition \eqref{condinvariance} is
satisfied, $\cA (R, [R, \xi_{e_1}]_\star)\approx 0$. It follows that
the same is true also for another element $R^\prime = N\star R=R\star N$ of
gauge $cu(2,2|\,8)$ superalgebra, \textit{i.e.} $\cA (N\star R, [N\star R,
\xi_{e_1}]_\star)\approx 0$. Since $N$ is central element of
$cu(2,2|\,8)$ and by virtue of the $C$-invariance condition one
obtains $\cA (R, [R, N\star N\star\xi_{e_1}])\approx 0$ for some new gauge
parameter $\zeta = N\star N\star\xi_{e_1}$. In fact, parameter $\zeta$ is a
combination of $\zeta_{e_1}$, $\zeta_{e_{31}}$, and $\zeta_{e_4}$,
expressed via $T^+$ and $T^-$ acting on original $\xi_{e_1}$. The
invariance with respect to $\xi_{e_1}$, $\xi_{e_{31}}$, and $\xi_{e_4}$ can also
be checked by direct calculation: varying with respect to
$\xi_{e_{31}}$ and $\xi_{e_{4}}$ gives the same relation between the respective
normalization constants  $\Phi_0$ as guaranteed by the $C$-invariance
condition \eqref{Cinvnumbers} and gives equations on  coefficient functions
equivalent to those that follow from the
variation with respect to $\xi_{e_{1}}$.

Analogous reasoning is also applied to the gauge transformations with
fermionic parameters $\xi_{o_1}$ and $\xi_{o_2}$ and it follows that
gauge invariance  $\xi_{o_2}$ is guaranteed by gauge invariance $\xi_{o_1}$
and the $C$-invariance condition. As a result, we obtain that it is sufficient to
check gauge invariance for three bosonic parameters $\xi_{e_1}$, $\xi_{e_{21}}$, $\xi_{e_{32}}{}_i{}^j$ and for
one fermionic parameter  $\xi_{o_{11}}{}_i$. Expression for these gauge transformations
are given in Appendix \bref{sec:A}. The invariance associated with other gauge parameters is guaranteed
through the $C$-invariance condition. In fact, imposing the gauge invariance with respect the above parameters
leaves just four independent constants $\Phi_0$ \eqref{solution} in front of actions $\cB_{1}$, $\cB_{2}$, $\cB_{32}$, and $\cF_1$. They will
be respectively denoted as $\Phi_0^{e_1}$, $\Phi_0^{e_2}$, $\Phi_0^{e_{32}}$, and $\Phi_0^{o_1}$.

Now we discuss the gauge invariance and linear relations on four normalization constants imposed
by each type of gauge symmetry. In order to find these relations
one needs to use the following identities between
functions $\Psi_0(X)$, $\Psi_\half(X)$, $\Psi_1(X)$ and their derivatives
with different values of a second spin
\be
\label{IdentitiesPhi}
\ba{c}
\dps
X\frac{\d \Psi_\half(X)}{\d X}+\Psi_\half(X)=\frac{1}{2}\,\Psi_0(X)\,,
\qquad
\frac{\d \Psi_{0} (X)}{\d X} = \frac{1}{2}\Psi_\half (X)\;,
\\
\\
\dps
X\frac{\d \Psi_{1}(X)}{\d X}  + 2 \Psi_{1}(X)  = \frac{1}{2}\Psi_{\half} (X)\;,
\qquad
\frac{\d \Psi_{\half} (X)}{\d X} = \frac{1}{2}\Psi_1 (X)\;.
\ea
\ee


Let us shortly discuss each of four types of gauge symmetry. Firstly, consider  gauge symmetry  with
parameter $\xi_{e_{21}} = \xi_{e_{21}}(a,b)$ and its conjugated cousin. Because  this
symmetry is bosonic it follows that fermionic and bosonic sectors of the full action \eqref{act34}
transform independently. In the fermionic sector the gauge symmetry
mixes up fields $\Omega_{o_1}$ and $\Omega_{o_2}$ \eqref{A5}, \eqref{A6} and by direct calculation one obtains
that fermionic sector is invariant provided
normalization constants are related as $\cF_1 = 4\cF_2$, cf. \eqref{Cinvnumbers}. In the bosonic sector
the gauge symmetry mixes up four fields $\Omega_{e_1}$, $\Omega_{e_{2}}$, $\Omega_{e_{31}}$,
and $\Omega_{e_4}$  \eqref{A7}. Calculating the respective action's variation,
using  identities
\eqref{IdentitiesPhi} and the $C$-invariance condition one obtains
\be
\Phi_0^{e_2} = 2\Phi_0^{e_1}\;,
\ee
while $\cB_{{31}} = \half \cB_{1}$ and $\cB_{{4}} = \frac{1}{4}\cB_{1}$. It follows that normalization
constants in this sector of fields are totally fixed in terms of $\Phi_0^{e_1}$.

Quite analogously we consider gauge symmetry with $su(2)$ matrix-valued parameter
$\xi_{e_{32}}{}_j^i = \xi_{e_{32}}{}_j^i(a,b)$. Since this
symmetry is bosonic it follows that fermionic and bosonic sectors of the full action \eqref{act34}
transform independently. In the fermionic sector the gauge symmetry
mixes up fields $\Omega_{o_1}$ and $\Omega_{o_2}$ \eqref{A8} and by direct calculation one obtains
that fermionic sector is invariant provided
normalization constants are related as $\cF_1 = 4\cF_2$, cf. \eqref{Cinvnumbers}.
In the bosonic sector
the gauge symmetry mixes up four fields $\Omega_{e_1}$, $\Omega_{e_{32}}{}_i{}^j$,
and $\Omega_{e_4}$  \eqref{A9}. Calculating the respective action's variation, using identities
\eqref{IdentitiesPhi} and the $C$-invariance condition one obtains
\be
\Phi_0^{e_{32}} = \frac{1}{4}\Phi_0^{e_1}\;,
\ee
while $\cB_{{4}} = \frac{1}{4}\cB_{1}$. It follows that normalization
constants in this sector of fields are completely  fixed in terms of $\Phi_0^{e_1}$. It also implies
that all bosonic coefficients are fixed uniquely and the overall normalization constant is $\Phi_0^{e_1}$.

Finally, we analyze fermionic $su(2)$ vector-valued parameter $\xi_{o_{11}}^i = \xi_{o_{11}}^i(a,b) $ and
its conjugated one. The respective gauge transformation is supersymmetric and mixes up all bosonic fields
and all fermionic fields, see \eqref{A10}. Calculating the respective action's
variation, using  identities
\eqref{IdentitiesPhi} and the $C$-invariance condition one obtains
\be
\Phi_0^{o_1} = - \Phi_0^{e_1}\;,
\ee
$\cB_{{31}} = \half \cB_{1}$, $\cB_{{4}} = \frac{1}{4}\cB_{1}$, and $\cF_1 = 4\cF_2$, cf. \eqref{Cinvnumbers}.
It follows that all normalization constants are fixed uniquely and expressed in terms of $\Phi_0^{e_1}$
to be denoted as
\be
\Phi_0 \equiv  \Phi_0^{e_1}\;.
\ee
The final expressions for coefficient functions are collected in Section \bref{sec:summary}.

\section{Conclusion}
\label{sec:conclusion}

In this paper we  built and analyzed FV-type  formulation of
$\ads$ totally symmetric and mixed-symmetry massless fields
interacting between themselves and with the gravity. Our
consideration is performed in the cubic order approximation. We
considered two models with gauge symmetry corresponding to
reduced and unreduced  $\cN=2$ Fradkin-Linetsky higher spin
superalgebras, $cu(2,2|8)$ and $hu_0(2,2|8)$. We have built the
projecting operator that explicitly factorizes unreduced
superalgebra $cu(2,2|8)$ to obtain reduced  superalgebra
$hu_0(2,2|8)$. Moreover, we have found projecting operators for
any $\cN$.

It is worth noting that constructing the interaction vertices brings to light
very powerful algebraic tools like Howe dual pairs of classical Lie
(super)algebras realized on a superspace of auxiliary variables.
One of the most important implications of Howe duality  is the $gl(1)$ invariance condition
referred to as the $C$-invariance condition for the action functional
\eqref{Cinv}. This condition is the direct  analog of the $sp(2)$
invariance for  Vasiliev equations for totally symmetric fields
\cite{Vasiliev:2003ev}. Indeed,  $N$ is  the basis element of
$gl(1)$   considered as Howe dual algebra to
$su(2,2\,|\, 2)$ superalgebra  in the star product realization.
Then the condition $[N, F]_\star = 0$ \eqref{dpoc} tells us that
fields are $gl(1)$ invariants and this invariance should be
retained on the action level via the $C$-invariance condition.

Let us now discuss some future research directions. First of all,
it would be worth pursuing our analysis to $\cN>2$ thereby
including mixed-symmetry fields of any value of the second spin
$s_2$ and not only "hook" fields with $s_2=1$. Further progress
depends on establishing for spins $s_2>1$ the proposition
analogous to those of Section \bref{sec:weyl}. Namely, it is
necessary to formulate a proper set of constraints for unfolded
fields such that one obtains correct on-shell dynamics. We hope to
return to this problem elsewhere.

Much more important and difficult  task however is to construct
nonlinear equations of motion for mixed-symmetry fields in all
orders thereby extending Vasiliev equations for totally symmetric
fields \cite{Vasiliev:2003ev}. Contrary to the on-shell theory one
may consider also the so-called off-shell formulation of higher
spin dynamics that introduces higher spin fields and their
non-linear gauge symmetries without imposing any field equations.
It will be interesting to develop the off-shell nonlinear
formulation for mixed-symmetry fields both in Minkowski and AdS
spacetimes as it has been done in the case of totally symmetric
fields \cite{Sagnotti:2005ns,Vasiliev:2005zu,Grigoriev:2006tt}.

It would be useful  to extend results of the present paper to
higher dimensions $d>5$ and consider a FV-type theory based on the
higher spin algebra $hu(1|(1, 2):[M, 2])$ from
\cite{Vasiliev:2004cm}. Gauging this algebra yields generalized
"hook" massless fields in $AdS_d$ spacetime, which are fields with
one row of any length  and one column of any height (in fact, the
height is bounded from below by a dimension $d$).

\vspace{5mm}

\textbf{Acknowledgments.} I am grateful to Maxim Grigoriev, Ruslan
Metsaev, and  Eugene  Skvortsov for  many fruitful discussions.
Also I would like to thank Harald Dorn for his  hospitality during
my visit to Humboldt university, Berlin, where this work has been
finished. This work is supported by RFBR grant Nr 08-02-00963 and
the Alexander von Humboldt Foundation grant PHYS0167.


\subsection*{\textbf{\Large Appendix}}
\refstepcounter{section}
\label{sec:A}
\addcontentsline{toc}{subsection}{Appendix}
\def\theequation{A.\arabic{equation}}
\setcounter{equation}{0}

\paragraph{Coefficient functions.}

Spin-$(s,0)$ case, see \cite{Vasiliev:2001wa}:
\be
\label{fron}
\Phi(X) =  - X (\alpha(X,-X) - 2\gamma(X,-X))\;.
\ee
Spin-$(s,\half)$ case, see \cite{Alkalaev:2002rq}:
\be
\label{fronfng}
\Phi(X) = X(\alpha(X,-X) + \zeta(X,-X))\;.
\ee
Spin-$(s,1)$ case, see \eqref{ExampleOfPhi}:
\be
\label{hook}
\Phi(X) = X(\alpha(X,-X) - \zeta(X,-X))\;.
\ee


In what follows we list  explicit expressions for gauge transformations.
We use commutators $[F,G]_\star  = F\star G-G\star F$ and anticommutators
$\{F,G\}_\star =F\star G +G\star F$.

\paragraph{The gauge symmetry with parameter $\xi_{e_1}(a,b)$.}

\be
\label{purebosonicGTR}
\ba{c}
\dps
\delta R_{e_{1}} = [R_{e_1},\xi_{e_1}]_\star\;,
\qquad
\delta R_{e_{4}} = [R_{e_4},\xi_{e_1}]_\star\;,
\\
\\
\delta R_{e_{21}} = [R_{o_{21}},\xi_{e_1}]_\star \;,
\qquad
\delta R_{e_{22}} = [R_{o_{22}},\xi_{e_1}]_\star \;,
\\
\\
\delta R_{e_{31}} = [R_{e_{31}},\xi_{e_1}]_\star \;,
\qquad
\delta R_{e_{32}}{}^i{}_j = [R_{e_{32}}{}^i{}_j,\xi_{e_1}]_\star \;,

\\
\\
\delta R_{o_{11}}^i = [R_{o_{11}}^i,\xi_{e_1}]_\star \;,
\qquad
\delta R_{o_{12}}{}_i = [R_{o_{12}}{}_i,\xi_{e_1}]_\star \;,
\\
\\
\delta R_{o_{21}}^i = [R_{o_{21}}^i,\xi_{e_1}]_\star \;,
\qquad
\delta R_{o_{22}}{}_i = [R_{o_{22}}{}_i,\xi_{e_1}]_\star \;,
\ea
\ee

\paragraph{The gauge symmetry with parameter $\xi_{e_{21}}(a,b)$ in the fermionic sector.}

\be
\label{A5}
\ba{c}
\dps
\delta R_{o_{11}}^i = -\epsilon^{ij}\{R_{o_{12}{}_j}, \xi_{e_{21}}\}_\star
-\half \epsilon^{ij}[R_{o_{22}{}_j}, \xi_{e_{21}}]_\star\;,
\qquad
\delta R_{o_{12}}{}_i = 0 \;,
\ea
\ee
and
\be
\label{A6}
\ba{c}
\dps
\delta R_{o_{21}}^i = 2 \epsilon^{ij}[R_{o_{12}{}_j}, \xi_{e_{21}}]_\star
+ \epsilon^{ij}\{R_{o_{22}{}_j}, \xi_{e_{21}}\}_\star\;,
\qquad
\delta R_{o_{22}}{}_i =0\;.
\ea
\ee
The analogous transformations hold for the conjugated gauge parameter $\xi_{e_{22}}$.

\paragraph{The gauge symmetry with parameter  $\xi_{e_{21}}(a,b)$ in the bosonic sector.}

\be
\label{A7}
\ba{c}
\delta R_{e_1} = -[R_{e_{22}}, \xi_{e_{21}}]_\star\;,
\\
\\
\dps
\delta R_{e_{21}} = [R_{e_{1}}, \xi_{e_{21}}]_\star+\{R_{e_{31}}, \xi_{e_{21}} \}_\star
+\half [R_{e_4},\xi_{e_{21}}]_\star\;,
\\
\\
\dps
\delta R_{e_{22}} = 0\;,
\qquad
\delta R_{e_{32}}{}^i{}_j = 0\;,
\\
\\
\dps
\delta R_{e_{31}} = 2 \{R_{e_{22}}, \xi_{e_{21}}\}_\star\;,
\qquad
\delta R_{e_4} = -2 [R_{e_{22}},\xi_{e_{21}}]_\star\;.
\ea
\ee

\paragraph{The gauge symmetry with parameter $\xi_{e_{32}}(a,b)$ in the fermionic sector.}
The symmetry associated with parameter
$\xi_{e_{32}}(a,b, \psi) = \xi_{e_{32}}{}^i_j(a,b)(\psi_i\bpsi^j)$, where
we assume that all $su(2)$ traces are zero, has the following form
\be
\label{A8}
\ba{c}
\dps
\delta R_{o_{11}}^ i = - \frac{1}{2}\{ R^j_{o_{11}}, \xi_j{}^i\}_\star -\frac{1}{4}[R_{o_{21}}^j, \xi_j{}^i]_\star \;,
\\
\\
\dps
\delta R_{o_{12}}{}_i = \frac{1}{2}\{ R_{o_{12}}{}_j, \xi^j{}_i\}_\star - \frac{1}{4}[ R_{o_{22}}{}_j, \xi^j{}_i]_\star\;,
\\
\\
\dps
\delta R_{o_{21}}^i = -[R_{o_{11}}^j, \xi_j{}^i]_\star - \frac{1}{2} \{ R^j_{o_{21}}, \xi_j{}^i\}_\star  \;,
\\
\\
\dps
\delta R_{o_{22}}{}_i = -[R_{o_{12}}{}_j, \xi^j{}_i]_\star +\frac{1}{2}\{ R_{o_{22}}{}_j, \xi^j{}_i\}_\star\;.
\ea
\ee

\paragraph{The gauge symmetry for $\xi_{e_{32}}(a,b)$ in the bosonic sector.}
The symmetry associated with parameter
$\xi_{e_{32}}(a,b, \psi) = \xi_{e_{32}}{}^i_j(a,b)(\psi_i\bpsi^j)$, where
we assume that all $su(2)$ traces are zero, has the following form
\be
\label{A9}
\ba{l}
\dps
\delta R_{e_1} = \frac{1}{4}[R_{e_{32}}{}^m{}_n,  \xi_{e_{32}}{}^n{}_m]_\star \;,
\qquad
\delta R_{e_{21}}= 0\;,
\qquad
\delta R_{e_{31}}  = 0\;,
\\
\\
\dps
\delta R_{e_{32}}{}^i{}_j = [R_{e_1}, \xi_{e_{32}}{}^i_j]_\star
-\frac{1}{2}[R_{e_4},\xi_{e_{32}}{}^i{}_j ]_\star +
\\
\\
\dps
+\frac{1}{2}\big(\{R_{e_{32}}{}^i{}_n,  \xi_{e_{32}}{}^n{}_j\}_\star
- \frac{1}{2} \delta^i_j\,\{R_{e_{32}}{}^m{}_n,  \xi_{e_{32}}{}^n{}_m\}_\star\big)-
\\
\\
\dps
-\frac{1}{2}\big(\{R_{e_{32}}{}^m{}_j,  \xi_{e_{32}}{}^i{}_m\}_\star
- \frac{1}{2} \delta^i_j\,\{R_{e_{32}}{}^m{}_n,  \xi_{e_{32}}{}^n{}_m\}_\star\big)\;,
\\
\\
\dps
\delta R_{e_4} = -\frac{1}{2}[R_{e_{32}}{}^m{}_n, \xi_{e_{32}}{}^n{}_m]_\star\;.
\ea
\ee

\paragraph{Supersymmetry transformations.}
Let us choose supersymmetric  parameter in the form $\xi_{o_{12}} = \xi_i(a,b)\bpsi^i$.
\be
\nonumber
\ba{l}
\dps
\delta R_{e_1} = \frac{1}{2}[R_{o_{11}}^i, \xi_i]_\star\;,
\qquad
\delta R_{e_4}  = \frac{1}{2}\{R_{o_{21}}^i, \xi_i\}_\star\;,
\\
\\
\dps
\delta R_{o_{11}}^i = \epsilon^{ij}\{R_{e_{21}}, \xi_j\}_\star\;,
\qquad
\delta R_{o_{12}}{}_i  =  [R_{e_1}, \xi_i]_\star- \frac{1}{2}\{R_{e_{32}}{}^m{}_i, \xi_m\}_\star
- \frac{1}{2}\{R_{e_{31}}, \xi_i\}_\star\;,
\\
\\
\dps
\delta R_{e_{21}} = 0\;,
\qquad
\delta R_{e_{22}}  = \frac{1}{2}\epsilon^{ij}\{R_{o_{12}}{}_i, \xi_j\}_\star
-\frac{1}{4}\epsilon^{ij} [R_{o_{22}}{}_i, \xi_j]_\star \;,

\ea
\ee
\be
\label{A10}
\ba{l}
\dps
\delta R_{e_{31}} = \frac{1}{4}[R_{o_{21}}^i, \xi_i]_\star
+\frac{1}{2}\{R_{o_{11}}^m, \xi_m\}_\star\;,

\\
\\
\dps
\delta R_{e_{32}}{}^i{}_j  = \{R_{o_{11}}^i, \xi_j\}_\star
-\frac{1}{2}\big([R_{o_{21}}^i, \xi_j]_\star - \frac{1}{2}\delta^i_j\,[R_{o_{21}}^m, \xi_m]_\star\big)\;,
\\
\\
\dps
\delta R_{o_{21}}^i  = 2\epsilon^{ij} [R_{e_{21}}, \xi_j]_\star
\;,
\qquad
\delta R_{o_{22}}{}_i  =
[R_{e_{31}}, \chi_i]_\star -[R_{e_{32}}{}^m{}_i, \xi_m]_\star - [R_{e_4}, \xi_i]_\star \;.

\ea
\ee

\vspace{5mm}



\providecommand{\href}[2]{#2}\begingroup\raggedright\endgroup

\end{document}